\documentclass[12pt]{article}

\usepackage[paper=letterpaper,margin=1in]{geometry}

\usepackage{thm-restate}
\usepackage{booktabs}

\title{Object Allocation Over a Network of Objects:\\ 
	Mobile Agents with Strict Preferences\footnote{	Department of Computer Science,
	University of Texas at Austin. A preliminary version of this paper will appear as an extended abstract in the Proceedings of the 20th International Conference on Autonomous Agents and Multiagent Systems. Email: \{fuli, plaxton, vaibhavsinha\}@utexas.edu.
}
}

\author{
	Fu Li
	\and
	C.\ Gregory Plaxton
	\and 
	Vaibhav B. Sinha
}

\date{March 2021}

\usepackage[ruled,vlined]{algorithm2e}
\usepackage[utf8]{inputenc}
\usepackage{amsmath,amsfonts,amsthm,amssymb}
\usepackage{graphicx}
\usepackage{tikz}

\usepackage{wrapfig}

\newcommand{\punt}[1]{#1}

\newcommand{\Agent}{a}
\newcommand{\Agents}{A}
\newcommand{\Object}{b}
\newcommand{\Objects}{B}
\newcommand{\Prefs}{\succ}
\newcommand{\Oaf}{F}
\newcommand{\Edge}{e}
\newcommand{\Edges}{E}
\newcommand{\Matching}{\mu}
\newcommand{\Swap}[3]{#2\rightarrow_{#1}#3}
\newcommand{\Swaps}[3]{#2\leadsto_{#1}#3}
\DeclareMathOperator{\ReachOp}{reach}
\newcommand{\Config}{\chi}
\newcommand{\Reach}[1]{\ReachOp(#1)}
\newcommand{\ReachAgain}[2]{\ReachOp(#1,#2)}

\DeclareMathOperator{\AgentsOp}{agents}
\newcommand{\AgentsOf}[1]{\AgentsOp(#1)}
\DeclareMathOperator{\ObjectsOp}{objects}
\newcommand{\ObjectsOf}[1]{\ObjectsOp(#1)}

\newcommand{\Weight}[1]{w(#1)}
\DeclareMathOperator{\LeftOp}{left}
\newcommand{\Left}[2]{\LeftOp(#1,#2)}
\DeclareMathOperator{\RightOp}{right}
\newcommand{\Right}[2]{\RightOp(#1,#2)}

\newcommand{\MaxMatched}[1]{\max(#1)}
\DeclareMathOperator{\MinUnmatchedOp}{hole}
\newcommand{\MinUnmatched}[1]{\MinUnmatchedOp(#1)}

\newcommand{\Congruent}{\simeq}

\newcommand{\MatchingsSymbol}{\Phi}
\newcommand{\Mcms}[1]{\MatchingsSymbol(#1)}
\newcommand{\Decs}[1]{\MatchingsSymbol^*(#1)}
\newcommand{\DecsBinary}[1]{\MatchingsSymbol'(#1)}

\DeclareMathOperator{\SpanOp}{span}
\newcommand{\Span}[2]{\SpanOp(#1,#2)}
\DeclareMathOperator{\IndRatOp}{IR}
\newcommand{\IndRat}[2]{\IndRatOp(#1,#2)}
\newcommand{\IndRatAgent}[3]{\IndRatOp(#1,#2,#3)}

\newcommand{\Encs}[2]{\Psi(#1,#2)}
\newcommand{\Steps}[2]{{#1}\leadsto{#2}}

\DeclareMathOperator{\DestsOp}{indices}
\newcommand{\Dests}[3]{\DestsOp(#1,#2,#3)}
\newcommand{\BranchObject}[2]{\langle #1,#2\rangle}
\newcommand{\Center}{o}
\newcommand{\Partial}{\tau}
\newcommand{\AgentSeq}{\sigma}
\DeclareMathOperator{\SerialOp}{serial}
\newcommand{\Serial}[1]{\SerialOp(#1)}

\newcommand{\LLL}{\mathcal{L}}
\newcommand{\Instance}{I}
\newcommand{\Improving}[2]{\ObjectsOp(#1,#2)}

\newcommand{\fu}[1]{}

\newtheorem{theorem}{Theorem}[section]
\newtheorem{lemma}[theorem]{Lemma}

\newtheorem{ob}[theorem]{Observation}
\newtheorem{claim}[theorem]{Claim}

\newcommand{\cutoff}[1]{ }

\newcommand{\RO}{\textsf{Reachable Object}}
\newcommand{\RM}{\textsf{Reachable Matching}}
\newcommand{\PE}{\textsf{Pareto Efficiency}}

\DeclareMathOperator{\rank}{rank}

\begin{document}
	
	\begin{titlepage}
	\maketitle
        \thispagestyle{empty}

\begin{abstract}
In recent work, Gourv\`es, Lesca, and Wilczynski propose a variant of the classic housing markets model where the matching between agents and objects evolves through Pareto-improving swaps between pairs of adjacent agents in a social network.  To explore the swap dynamics of their model, they pose several basic questions concerning the set of reachable matchings.  In their work and other follow-up works, these questions have been studied for various classes of graphs: stars, paths, generalized stars (i.e., trees where at most one vertex has degree greater than two), trees, and cliques. For generalized stars and trees, it remains open whether a Pareto-efficient reachable matching can be found in polynomial time.

In this paper, we pursue the same set of questions under a natural variant of their model. In our model, the social network is replaced by a network of objects, and a swap is allowed to take place between two agents if it is Pareto-improving and the associated objects are adjacent in the network. In those cases where the question of polynomial-time solvability versus NP-hardness has been resolved for the social network model, we  are able to show that the same result holds for the network-of-objects model.  In addition, for our model, we present a polynomial-time algorithm for computing a Pareto-efficient reachable matching in generalized star networks. Moreover, the object reachability algorithm that we present for path networks is significantly faster than the known polynomial-time algorithms for the same question in the social network model.
\end{abstract}

	\end{titlepage}


\section{ Introduction}
\label{sec:intro}

Problems related to resource allocation under preferences are widely
studied in both computer science and economics. Research in this area
seeks to gain mathematical insight into the structure of resource
allocation problems, and to exploit this structure to design fast
algorithms. In one important class of resource allocation problems,
sometimes referred to as one-sided matching
problems~\cite{david2013algorithmics}, we seek to allocate indivisible
objects to a set of agents, where each agent has preferences over the
objects and wants to receive at most one object (unit demand). The
allocation should enjoy one or more strong game-theoretic properties,
such as Pareto-efficiency.

In a seminal work, Shapley and Scarf~\cite{shapley1974cores}
introduced the notion of a housing market, which corresponds to the
special case of one-sided matching in which there are an equal number
of agents and objects, each agent is initially endowed with a distinct
object, and each agent is required to be matched to exactly one
object.  They present an elegant algorithm (attributed to David Gale)
for housing markets called the top trading cycles (TTC) algorithm. The
TTC algorithm enjoys a number of strong game-theoretic properties. 
For example, when agents have strict preferences, the output of the TTC
algorithm is the unique matching in the core. The TTC algorithm has
subsequently been generalized to handle more complex variants of the
original housing market problem (e.g., \cite{aziz2016optimal, chevaleyre2007allocating,endriss2006negotiating,saad2009coalitional,  sandholm1998contract}). 

Like many one-sided matching algorithms, the TTC algorithm is
centralized: it takes all of the agent preference information as input
and computes the output matching.  In some resource allocation
scenarios of practical interest, it may be difficult or impossible to
coordinate such a global recomputation of the matching.  Accordingly,
researchers have studied decentralized (or distributed) variants of
one-sided matching problems 
in which the initial allocation gradually evolves as ``local'' trading
opportunities arise. 
In this setting, restrictions are imposed on the sets of agents
that are allowed to participate in a single trade. 
For example, we might only allow (certain) pairs of agents to trade.
In addition, all trades are required to be Pareto-improving.
Locality-based restrictions on trade are generally
enforced through graph-theoretic constraints.

Of particular relevance to the present paper is the line of research
initiated by Gourv{\`e}s et al.~\cite{gourves2017object} on
decentralized allocation in housing markets. They propose a model in
which agents have strict preferences and are embedded in an underlying
social network. A pair of agents are allowed to swap objects with each
other only if (1) they will be better off after the swap, and (2) they
are directly connected (socially tied) via the network. The underlying
social network is modeled as an undirected graph, and five different
graph classes are considered: paths, stars, generalized stars, trees,
and general graphs.  The swap dynamics of the model are investigated
by considering three computational questions. The first question,
\RO, asks whether there is a sequence of swaps that results in a
given agent being matched to a given target object. The second
question, \RM, asks whether there is a sequence of swaps that
results in a given target matching. The third question, \PE, asks
how to find a sequence of swaps that results in a Pareto-efficient
matching with respect to the set of reachable matchings.

Gourv{\`e}s et al.~\cite{gourves2017object} studied each of the three
questions in the context of the aforementioned graph classes, with the
goal of either exhibiting a polynomial-time algorithm or establishing
NP-hardness. 
For some of these problems, it is a relatively straightforward exercise to
design a polynomial-time algorithm (even for the search version). In
particular, this is the case for all three reachability questions on
stars, for \PE\ on paths, and for \RM\ on trees (which subsumes \RM\ 
on generalized stars, and hence also on paths).  Gourv{\`e}s et al.\ present an
elegant reduction from 2P1N-SAT~\cite{Yoshinaka05} to establish the NP-completeness of
\RO\ on generalized stars (and hence also on trees and general graphs). They establish the
NP-completeness of \RM\ on general graphs via a reduction from \RO\ on
trees. The latter reduction has the property that for any given
instance of \RO\  on trees, the target matching in the instance of \RM\ 
on general graphs produced by the transformation matches each agent to
its most preferred object. Consequently, the same reduction
establishes the NP-hardness of \PE\ on general graphs.
The work of Gourv{\`e}s et al.\ left three of
these problems open: \RO\ on paths and \PE\ on generalized stars and
trees.  Subsequently, two sets of authors independently presented
polynomial-time algorithms for \RO\ on
paths~\cite{bentert2019good, huang2019object}. Both groups obtained an
$O(n^4)$-time algorithm by carefully studying the structure of swap
dynamics on paths and then reducing the problem to 2-SAT.
The complexity of \PE\ remains open
for generalized stars and for trees.
Gourv{\`e}s et al.\  noted, 
``It appears interesting to see if Pareto 
(Efficiency) is polynomial time solvable in a 
generalized star by a combination of the 
techniques used to solve the cases of paths 
and stars.''

Bentert et al.~\cite{bentert2019good} established that \RO\ on cliques
is NP-complete, and M{\"u}ller and Bentert~\cite{muller2020reachable}
established that \RM\ on cliques is NP-complete. It is easy to
extend the latter result to show that \PE\ on cliques is NP-hard.
These three hardness results for cliques subsume the corresponding
results obtained previously for general graphs by Gourv{\`e}s et al.

We study a natural variant of the
decentralized housing markets model of Gourv{\`e}s et
al.~\cite{gourves2017object}.  Instead of enforcing locality
constraints on trade via a network where the locations of the agents
are fixed (since they correspond to the vertices of the network) and
the objects move around (due to swaps), we consider a network where
the locations of the objects are fixed and the agents move around. We
refer to these two models as the object-moving model and the
agent-moving model.
Table~\ref{tbl:dozen} summarizes the current state of the art for the object-moving model.

 To motivate the study of
the agent-moving model, consider a cloud computing environment with a
large number of servers (objects) connected by a network that are
available to rent. A set of customers (agents) are each interested in
renting one server.  The servers vary in CPU capacity, storage
capacity, physical security, and rental cost. Varying customer
workloads and requirements result in varying customer preferences over
the servers. Rather than attempting to globally optimize the entire
matching of customers to servers, it might be preferable to allow
local swaps between adjacent servers to gradually optimize the
matching.  Given that customer workloads are likely to vary
significantly over time, an optimization strategy based on frequent
local updates might outperform a strategy based on less frequent
global updates.  Alternatively, one can envision a system that
performs occasional global updates to optimize the matching, and that
relies on local updates to maintain a reasonable matching between
successive global updates.

{\bf Our Results.} We initiate the study of the agent-moving model by revisiting each of
the questions associated with Table~\ref{tbl:dozen} in the context of
the agent-moving model. 
We emphasize that the sole difference between
the agent-moving model and the object-moving model
is that the locality constraint prevents an agent $\Agent$ currently
matched to an object $\Object$ from trading with an agent $\Agent'$
currently matched to an object $\Object'$ unless objects $\Object$ and
$\Object'$ (two vertices in a given network of objects) are adjacent,
rather than requiring agents $\Agent$ and $\Agent'$ (two vertices in a
given network of agents) to be adjacent. Both models also require
swaps to be Pareto-improving.
The two models have strong similarities.
In fact, for all of the questions in Table
1 for which a polynomial-time algorithm or hardness result has been
established in the object-moving model, we establish a corresponding
result in the agent-moving model. Moreover, for Pareto Efficiency on
generalized stars, which is open in the object-moving model, we
provide a polynomial-time algorithm in the agent-moving model.  In
some cases, it is relatively straightforward to adapt known results
for the object-moving model to the agent-moving model. Below we
highlight our four main technical contributions, which address more challenging cases. 


\begin{table}
 \centering

 \label{tbl:dozen}

\scalebox{0.78}{
	\begin{tabular}{cccc}
		\toprule & \RO & \RM & \PE \\ 
		\midrule
		Star& poly-time & (poly-time) & poly-time\\ 
		Path& poly-time & (poly-time) & poly-time \\ 
		Generalized Star& NP-complete & (poly-time) & open\\ 
		Tree & (NP-complete) & poly-time & open\\ 
		Clique & NP-complete & NP-complete & NP-hard\\
		\bottomrule
	\end{tabular} 
}

 \caption{This table presents known complexity results for various questions related to the object-moving model of Gourv{\`e}s et
	al.~\cite{gourves2017object}. The results in parentheses follow directly from other table entries. 
	For the agent-moving model, we obtain the same results, except that we also give a polynomial-time algorithm for \PE\ on generalized stars. 
}
\end{table}


Our first main
technical result is an $O(n^2)$ time algorithm for \RO\ on paths in
the agent-moving model, which is much faster than the known
$O(n^4)$-time algorithms for \RO\ on paths in the object-moving
model. (Here $n$ denotes the number of agents/objects; the size of the
input is quadratic in $n$ since the preference list of each agent is
of length $n$.) The speedup is due to a simpler local characterization
of the reachable matchings on a path in the agent-moving model. 

In our
second main technical result, we obtain the same $O(n^2)$ time bound
for \PE\  on paths. Our algorithms for \RO\ and \PE\  are based on an
efficient subroutine for solving a certain constrained reachability
problem.  Roughly speaking, this subroutine determines all of the
possible matches for a given agent when certain agent-object pairs are
required to be matched to one another.  Our implementation involves a
trivial $O(n^2)$-time preprocessing phase followed by an $O(n)$-time
greedy phase. The preferences of the agents are only examined during
the preprocessing phase. The proof of correctness of the greedy phase
is somewhat nontrivial. We solve \RO\ on paths using a single
application of the subroutine, yielding an $O(n^2)$ bound.  Our
polynomial-time algorithm for \PE\ on paths uses $n$ applications of
our algorithm for \RO\ on paths. Since the preprocessing phase only
needs to be performed once, the overall running time remains $O(n^2)$.

In our third main technical result, we present a polynomial-time
algorithm for \PE\ on generalized stars, which remains open in the
object-moving model. 
To tackle this problem, 
we use the serial dictatorship algorithm
with the novel idea of dynamically choosing the dictator sequence. We
also leverage our techniques for solving \PE\ on paths.

The faster time bounds discussed above for the case of paths suggest
that the agent-moving model is simpler than the object-moving model,
at least from an upper bound perspective. Accordingly, we can expect
it to be a bit more challenging to establish the NP-completeness
results stated in Table~\ref{tbl:dozen} for the agent-moving model than for the
object-moving model.  In our fourth main technical result, we adapt an
NP-completeness proof developed by Bentert et al.~\cite{bentert2019good} in the context of
the object-moving model to the more challenging setting of the
agent-moving model. 
Specifically, we modify their reduction from
2P1N-SAT to establish that Reachable Object on cliques remains
NP-complete in the agent-moving model.



{\bf Related work.} For the object-moving model, Huang and
Xiao~\cite{huang2019object} study \RO\ with weak preferences, i.e.,
where an agent can be indifferent between different objects. Bentert
et al.~\cite{bentert2019good} establish NP-hardness for \RO\ on
cliques, and consider the case where the preference lists have bounded
length.
  Saffidine and
  Wilczynski~\cite{saffidine2018constrained} propose a variant of
  \RO\ where we ask whether a given agent is guaranteed to achieve 
  a specified level of satisfaction after any maximal sequence of rational exchanges. 
M{\"u}ller and
Bentert~\cite{muller2020reachable} study \RM\ on cliques and cycles.
Aspects related to social connectivity are also addressed in recent
work on envy-free allocations~\cite{beynier2019local,bilo2018almost}
and on trade-offs between efficiency and
fairness~\cite{igarashi2019pareto}. 

Our agent-moving model can be viewed as a game in which each agent
seeks to be matched to an object that is as high as possible on its
preference list.  If the game reaches a state in which no further
swaps can be performed, we say that an equilibrium matching has been
reached. Agarwal et al.~\cite{AgarwalEGV20} study a similar game
motivated by Schelling’s well-known residential segregation model. As
in our game, there are an equal number of agents and objects, the
objects correspond to the nodes of a graph, a matching is maintained
between the agents and the objects, and the matching evolves via
Pareto-improving, agent-moving swaps. There are also some significant
differences. In our model, each agent has static preferences over the
set of objects, and swaps can only occur between adjacent agents (i.e,
agents matched to adjacent objects).  In the Agarwal et al.\ game,
each agent has a type, the desirability of an object $\Object$ to an
agent $\Agent$ depends on the current fraction of agents in the
``neighborhood'' of $\Object$ (i.e., the set of agents matched to an
object adjacent to $\Object$) with the same type as $\Agent$, and
swaps can occur between any pair of agents. Agarwal et al.\ study the
existence, computational complexity, and quality of equilibrium
matchings in such games.  Bil{\`{o}} et.\ al \cite{BiloBLM20} further
investigated the influence of the graph structure on the resulting
strategic multi-agent system.

{\bf Organization of the paper.} The remainder of the paper is organized as
follows. Section~\ref{sec:prelims} provides formal
definitions. Section~\ref{sec:path} presents our polynomial-time
algorithms for \RO\ and \PE\ on paths. 
Section~\ref{sec:gsPeAlg} presents our polynomial-time
algorithm for \PE\ on generalized stars. 
Section~\ref{sec:NPC_RO_clique} presents our NP-completeness result
for \RO\ on cliques.
Section~\ref{sec:hardness}
presents our other NP-completeness and NP-hardness results.
Section~\ref{sec:easy} briefly discusses simple algorithms for
justifying the other polynomial-time entries in Table~\ref{tbl:dozen}.
Section~\ref{sec:concs} offers  concluding remarks.

	\section{Preliminaries}
\label{sec:prelims}

We define an object allocation framework (OAF) as a $4$-tuple
$\Oaf=(\Agents,\Objects,\Prefs,\Edges)$ where $\Agents$ is a set of
agents, $\Objects$ is a set of objects such that
$|\Agents|=|\Objects|$, $\Prefs$ is a collection of strict linear
orderings $\{\Prefs_\Agent\}_{\Agent\in\Agents}$ over $\Objects$ such
that $\Prefs_\Agent$ specifies the preferences of agent $\Agent$ over
$\Objects$, and $\Edges$ is the edge set of some undirected graph
$(\Objects,\Edges)$.

We define a matching $\Matching$ of given OAF
$\Oaf=(\Agents,\Objects,\Prefs,\Edges)$ as a subset of
$\Agents\times\Objects$ such that no agent or object belongs to more
than one pair in $\Matching$. (Put differently, $\Matching$ is a
matching in the complete bipartite graph of agents and objects.)  We
say that such a matching is perfect if $|\Matching|=|\Agents|$. For
any matching $\Matching$, we define $\AgentsOf{\Matching}$ (resp.,
$\ObjectsOf{\Matching}$) as the set of all matched agents (resp.,
objects) with respect to $\Matching$.  For any matching $\Matching$
and any agent $\Agent$ that is matched in $\Matching$, we use the
shorthand notation $\Matching(\Agent)$ to refer to the object matched
to agent $\Agent$.  For any matching $\Matching$ and any object
$\Object$ that is matched in $\Matching$, we use the notation
$\Matching^{-1}(\Object)$ to refer to the agent matched to object
$\Object$.

For any OAF $\Oaf=(\Agents,\Objects,\Prefs,\Edges)$, any perfect
matching $\Matching$ of $\Oaf$, and any edge
$\Edge=(\Object,\Object')$ in $\Edges$ such that
$\Object'\Prefs_\Agent\Object$ and $\Object\Prefs_{\Agent'}\Object'$
where $\Agent=\Matching^{-1}(\Object)$ and
$\Agent'=\Matching^{-1}(\Object')$, we say that a swap operation is
applicable to $\Matching$ across edge $\Edge$, and we write
$\Swap{\Oaf,\Edge}{\Matching}{\Matching'}$ where
\[
\Matching'=(\Matching\setminus\{(\Agent,\Object),(\Agent',\Object')\})\cup
\{(\Agent,\Object'),(\Agent',\Object)\},
\]
is the matching of $\Oaf$ that results from applying this operation.
We write $\Swap{\Oaf}{\Matching}{\Matching'}$ to denote that
$\Swap{\Oaf,\Edge}{\Matching}{\Matching'}$ for some edge $\Edge$.  We
write $\Swaps{\Oaf}{\Matching}{\Matching'}$ if there exists a sequence
$\Matching=\Matching_0,\ldots,\Matching_k=\Matching'$ of matchings of
$\Oaf$ such that $\Swap{\Oaf}{\Matching_{i-1}}{\Matching_i}$ for
$1\leq i\leq k$.

We define a configuration as a pair $\Config=(\Oaf,\Matching)$ where
$\Oaf$ is an OAF and $\Matching$ is a perfect matching of $\Oaf$.

For any configuration $\Config=(\Oaf,\Matching)$ where
$\Oaf=(\Agents,\Objects,\Prefs,\Edges)$, any agent $\Agent$ in
$\Agents$, and any object $\Object$ in $\Objects$, we define
$\Config(\Agent)$ as a shorthand for the object $\Matching(\Agent)$,
and we define $\Config^{-1}(\Object)$ as a shorthand for the agent
$\Matching^{-1}(\Object)$.

For any configuration $\Config=(\Oaf,\Matching)$ where
$\Oaf=(\Agents,\Objects,\Prefs,\Edges)$, and any matching $\Matching'$
of $\Oaf$ such that $\Swap{\Oaf,\Edge}{\Matching}{\Matching'}$ for
some edge $\Edge$ in $\Edges$, we say that a swap is applicable to
$\Config$ across edge $\Edge$, and the result of applying this
operation is the configuration $(\Oaf,\Matching')$.

For any configuration $\Config=(\Oaf,\Matching)$, we define
$\Reach{\Config}$ as the set of all perfect matchings $\Matching'$ of
$\Oaf$ such that $\Swaps{\Oaf}{\Matching}{\Matching'}$.  For any
configuration $\Config=(\Oaf,\Matching)$ and any matching $\Matching'$
of $\Oaf$, we define $\ReachAgain{\Config}{\Matching'}$ as the set of
all matchings $\Matching''$ in $\Reach{\Config}$ such that
$\Matching''$ contains $\Matching'$.

We now state the three reachability problems studied in this paper.
\begin{itemize}
\item The reachable matching problem: Given a configuration
$\Config=(\Oaf,\Matching)$ and a perfect matching $\Matching'$ of
$\Oaf$, determine whether $\Matching'$ belongs to $\Reach{\Config}$.

\item
The reachable object problem: Given a configuration
$\Config=(\Oaf,\Matching)$ where
$\Oaf=(\Agents,\Objects,\Prefs,\Edges)$, an agent $\Agent$ in
$\Agents$, and an object $\Object$ in $\Objects$, determine whether
there is a matching $\Matching'$ in $\Reach{\Config}$ such that
$\Matching'(\Agent)=\Object$.

\item
The Pareto-efficient matching problem: Given a configuration
$\Config$, find a matching in $\Reach{\Config}$ that is not
Pareto-dominated by any other matching in $\Reach{\Config}$.
\end{itemize}

	\section{Reachability Over a Path Network}
\label{sec:path}

%

We begin by introducing some notation.

For any nonnegative integer $n$, we define $[n]$ as $\{1,\ldots,n\}$.
Without loss of generality, in this section we restrict attention to
OAFs of the form $([n], [n], \Prefs, \{(b,b+1)\mid 1\leq b<n\})$ for
some positive integer $n$. We use the notation $(n,\Prefs)$ to refer
to such an OAF.

For any nonnegative integer $n$, we define $\Mcms{n}$ as the set of all
matchings $\Matching$ such that $\AgentsOf{\Matching}=[|\Matching|]$
and $\ObjectsOf{\Matching}\subseteq [n]$.

For any matching $\Matching$ in $\Mcms{n}$, we define
$\MaxMatched{\Matching}$ as the maximum matched object in
$\ObjectsOf{\Matching}$, or as $0$ if $\Matching=\emptyset$.

For any matching $\Matching$ in $\Mcms{n}$, we define
$\MinUnmatched{\Matching}$ as the minimum positive integer that does
not belong to $\ObjectsOf{\Matching}$.

For any matching $\Matching$ in $\Mcms{n}$ and any agent $\Agent$ in
$\AgentsOf{\Matching}$, define $\Span{\Matching}{\Agent}$ 
as $\{\Object\in [n]\mid\Matching(\Agent)\leq\Object\leq\Agent\}
\cup\{\Object\in [n]\mid\Agent\leq\Object\leq\Matching(\Agent)\}.$

For any OAF $\Oaf=(n,\Prefs)$, we define $\Matching_\Oaf$ as the
matching $\{(i,i)\mid i\in [n]\}$, and we define $\Config_\Oaf$ as the
configuration $(\Oaf,\Matching_\Oaf)$.

For any OAF $\Oaf=(n,\Prefs)$ and any agent $\Agent$ in $[n]$, we
define $\Left{\Prefs}{\Agent}$ as the minimum object $\Object$ in
$[n]$ such that
$\Object\Prefs_\Agent\Object+1\Prefs_\Agent\cdots\Prefs_\Agent\Matching_\Oaf(\Agent)=\Agent$,
and we define $\Right{\Prefs}{\Agent}$ as the maximum object $\Object$
in $[n]$ such that
$\Object\Prefs_\Agent\Object-1\Prefs_a\cdots\Prefs_\Agent\Agent$.
Thus if a matching $\Matching$ belongs to $\Reach{\Config_\Oaf}$, then
the match $\Matching(\Agent)$ of agent $\Agent$ is at least
$\Left{\Prefs}{\Agent}$ and at most $\Right{\Prefs}{\Agent}$,
regardless of the preferences of the remaining agents.

For any OAF $\Oaf=(n,\Prefs)$, any matching $\Matching$ in $\Mcms{n}$,
and any agent $\Agent$ in $\AgentsOf{\Matching}$, we say that the
predicate $\IndRatAgent{\Prefs}{\Matching}{\Agent}$ holds (where
``IR'' stands for ``individually rational'') if $\Left{\Prefs}{\Agent}\leq\Matching(\Agent)\leq\Right{\Prefs}{\Agent}.$
We say that the predicate $\IndRat{\Prefs}{\Matching}$ holds if
$\IndRatAgent{\Prefs}{\Matching}{\Agent}$ holds for all agents
$\Agent$ in $\AgentsOf{\Matching}$.

\subsection{A Useful Subroutine}
\label{sec:greedy}

This section presents Algorithm~\ref{alg:greedyPath}, a greedy
subroutine that we use in Sections~\ref{sec:pathReachObject}
and~\ref{sec:pathPareto} to solve reachability problems over a path
network.


Recall that the reachable object problem with path configuration
$\Config_\Oaf$ is to check whether an object $\Object$ is reachable
for an agent $\Agent$. Algorithm~\ref{alg:greedyPath} addresses a
variant of this problem in which the agents less than $\Agent$ are all
required to be matched to specific objects. The input matching
$\Matching_0$ specifies the required match for each of these agents.
 
\begin{algorithm}[h]
\SetAlgoLined
 
\KwIn{An OAF $\Oaf=(n,\Prefs)$, a matching $\Matching_0$ in
$\Mcms{n}$ such that $|\Matching_0|<n$ and
$\ReachAgain{\Config_\Oaf}{\Matching_0}\not=\emptyset$,
and a matching $\Matching_1=\Matching_0+(|\Matching_1|,\Object_0)$
where
$\MaxMatched{\Matching_0}<\Object_0\leq\Right{\Prefs}{|\Matching_1|}$}
\KwOut{A matching $\Matching$ in
$\ReachAgain{\Config_\Oaf}{\Matching_1}$, or $\emptyset$ if
this set is empty}
$\Matching=\Matching_1$\;
\While{$0<|\Matching|<n$}{
 \uIf{$\Left{\Prefs}{|\Matching|+1}\leq\MinUnmatched{\Matching}$}{
 $\Matching=\Matching+(|\Matching|+1,\MinUnmatched{\Matching})$\;
 }\uElseIf{$\MaxMatched{\Matching}<\Right{\Prefs}{|\Matching|+1}$}{
 $\Matching=\Matching+(|\Matching|+1,\MaxMatched{\Matching}+1)$\;
 }\Else{
 $\Matching=\emptyset$\;
 }
}
\Return $\Matching$
\caption{A greedy path reachability subroutine.}
\label{alg:greedyPath}
\end{algorithm}

\subsection{Proof of Correctness of Algorithm~\ref{alg:greedyPath}}
\label{sec:greedy-proof}

In this section, we establish the correctness of
Algorithm~\ref{alg:greedyPath}. 

We begin by defining a specific subset $\Decs{n}$ of $\Mcms{n}$. For
any matching $\Matching$ in $\Mcms{n}$ and any integer $i$ in
$[|\Matching|]$, let $\Matching_i$ be the matching such that
$\Matching_i\subseteq \Matching$ and $\AgentsOf{\Matching}=[i]$. Then
$\Decs{n}$ is the set of all matchings $\Matching$ such that
$\Matching$ belongs to $\Mcms{n}$ and for each $i$ in
$[|\Matching|-1]$, either $\Matching(i+1) =
\MinUnmatched{\Matching_i}$ or
$\MaxMatched{\Matching_i}<\Matching(i+1)\le n$.

\label{sec:structural}
We now present a number of useful structural properties of matchings in $\Decs{n}$.


\subsubsection{Representing a Matching as a Pair of Binary Strings}
\label{sec:pair}

For any binary string $\alpha$, we let $|\alpha|$ denote the length of
$\alpha$, and we let $\Weight{\alpha}$ denote the Hamming weight of
$\alpha$.  For any binary string $\alpha$ and any integer $i$ in
$[|\alpha|]$, we let $\alpha_i$ denote bit $i$ of $\alpha$.  For any
binary string $\alpha$, and any integers $i$ and $j$ in $[|\alpha|]$,
we let $\alpha_{i,j}$ denote the substring $\alpha_i\cdots\alpha_j$ of
$\alpha$.

For any integers $m$ and $n$ such that $0\leq m\leq n$, we let
$\Encs{m}{n}$ denote the set of all pairs of binary strings
$(\alpha,\beta)$ such that $|\alpha| = m$, $|\beta|=n$,
$\Weight{\alpha_{1,i}}\geq\Weight{\beta_{1,i}}$ holds for all $i$ in
$[m]$, $\Weight{\alpha}=\Weight{\beta}$, and $m<n$ implies
$\beta_n=1$.

For any $(\alpha,\beta)$ in $\Encs{m}{n}$, we define $[\alpha,\beta]$
as the cardinality-$m$ matching $\Matching$ in $\Mcms{n}$ constructed
as follows: for any agent $\Agent$ in $[m]$ such that $\alpha_\Agent$
is the $i$th $0$ (resp., $1$) in $\alpha$, we define
$\Matching(\Agent)$ as the index of the $i$th $0$ (resp., $1$) in
$\beta$.

\begin{ob}\label{ob:encode}
Let $(\alpha,\beta)$ belong to $\Encs{m}{n}$, let $\Matching$ denote
$[\alpha,\beta]$, and let $\Object$ belong to $[n]$. If $\Object$ is
unmatched in $\Matching$, then $\beta_\Object=0$. Otherwise, the
following conditions hold, where $\Agent$ denotes
$\Matching^{-1}(\Object)$: $\Agent>\Object$ implies
$\alpha_\Agent=\beta_\Object=0$, $\Agent<\Object$ implies
$\alpha_\Agent=\beta_\Object=1$, and $\Agent=\Object$ implies
$\alpha_\Agent=\beta_\Agent$.
\end{ob}

\begin{ob}\label{ob:extend}
Let $(\alpha,\beta)$ belong to $\Encs{m}{n}$ and let $\Matching$
denote $[\alpha,\beta]$. If $m<n$ then $(\alpha 0,\beta)$ belongs to
$\Encs{m+1}{n}$ and $[\alpha
  0,\beta]=\Matching+(m+1,\MinUnmatched{\Matching})$. Furthermore, for
any nonnegative integer $k$, $(\alpha 1,\beta 0^k1)$ belongs to
$\Encs{m+1}{n+k+1}$ and $[\alpha 1,\beta 0^k1]=\Matching+(m+1,n+k+1)$.
\end{ob}

For any $(\alpha,\beta)$ in $\Encs{m}{n}$, and any agent $\Agent$ in
$[m]$ such that $\alpha_\Agent=\beta_\Agent$ and
$\Weight{\alpha_{1,\Agent}}=\Weight{\beta_{1,\Agent}}$, we say that a
complement operation is applicable to $(\alpha,\beta)$ at agent
$\Agent$. The result of applying this operation is the pair of binary
strings $(\alpha',\beta')$ that is the same as $(\alpha,\beta)$ except
$\alpha'_\Agent=\beta'_\Agent=1-\alpha_\Agent$.  It is easy to
see that $(\alpha',\beta')$ belongs to $\Encs{m}{n}$.

For any $(\alpha,\beta)$ and $(\alpha',\beta')$ in $\Encs{m}{n}$, we
write $(\alpha,\beta)\Congruent(\alpha',\beta')$ to denote that
$(\alpha,\beta)$ can be transformed into $(\alpha',\beta')$ via a
sequence of complement operations.

\begin{ob}\label{ob:congruent}
Let $(\alpha,\beta)$ and $(\alpha',\beta')$ belong to
$\Encs{m}{n}$. Then $[\alpha,\beta]=[\alpha',\beta']$ if and only if
$(\alpha,\beta)\Congruent(\alpha',\beta')$.
\end{ob}

\begin{ob}\label{ob:prefix}
Let $(\alpha,\beta)$ belong to $\Encs{m}{n}$, and let
$(\alpha',\beta')$ belong to $\Encs{m'}{n'}$ where $m\leq m'$ and
$n\leq n'$.  Then $[\alpha,\beta]\subseteq[\alpha',\beta']$ if and
only if
$(\alpha,\beta)\Congruent(\alpha'_{1,|\alpha|},\beta'_{1,|\beta|})$.
\end{ob}

For any $(\alpha,\beta)$ in $\Encs{m}{n}$, and any object $\Object$ in
$[n-1]$ such that $\beta_\Object=1$ and $\beta_{\Object+1}=0$, we say
that a sort operation is applicable to $(\alpha,\beta)$ across
objects $\Object$ and $\Object+1$.  The result of applying this
operation is the pair of binary strings $(\alpha,\beta')$ that is the
same as $(\alpha,\beta)$ except $\beta'_\Object=0$ and
$\beta'_{\Object+1}=1$.

\begin{ob}\label{ob:sort}
Let $(\alpha,\beta)$ belong to $\Encs{m}{n}$ and let $(\alpha,\beta')$
be the result of applying a sort operation to $(\alpha,\beta)$ across
objects $\Object$ and $\Object+1$. Then $(\alpha,\beta')$ belongs to
$\Encs{m}{n}$.  Furthermore, if $m=n$ then
\[
[\alpha,\beta']=(\Matching
\setminus\{(\Agent,\Object),(\Agent',\Object+1)\})
\cup\{(\Agent',\Object),(\Agent,\Object+1)\}
\]
where $\Matching$ denotes $[\alpha,\beta]$, $\Agent$ denotes
$\Matching^{-1}(\Object)$, and $\Agent'$ denotes
$\Matching^{-1}(\Object+1)$.
\end{ob}

For any $(\alpha,\beta)$ and $(\alpha',\beta')$ in $\Encs{m}{n}$, we
write $\Steps{(\alpha,\beta)}{(\alpha',\beta')}$ to denote that
$(\alpha,\beta)$ can be transformed into $(\alpha',\beta')$ via a
sequence of complement and sort operations.

\begin{ob}\label{ob:encodings}
Let $\alpha$ be a binary string of length $m$, and let $\beta$ be a
binary string of length $n$ such that $m\leq n$.  Then
$(\alpha,\beta)$ belongs to $\Encs{m}{n}$ if and only if
$\Steps{(0^m,0^n)}{(\alpha,\beta)}$.
\end{ob}

For any nonnegative integer $n$, we define $\DecsBinary{n}$ as the set of all
matchings $\Matching$ in $\Mcms{n}$ such that
$\Matching=[\alpha,\beta]$ for some
$(\alpha,\beta)$ in $\Encs{|\Matching|}{\MaxMatched{\Matching}}$.

For any $(\alpha,\beta)$ in $\Encs{n}{n}$, and any agent $\Agent$ in
$[n-1]$ such that
$\Weight{\alpha_{1,\Agent}}>\Weight{\beta_{1,\Agent}}$ and
$\alpha_\Agent=1$, we say that a pivot operation is applicable to
$(\alpha,\beta)$ at agent $\Agent$.  The result of applying this
operation is the pair of binary strings $(\alpha',\beta)$ that is the
same as $(\alpha,\beta)$ except $\alpha'_\Agent=0$ and
$\alpha'_{\Agent'}=1$, where $\Agent'$ denotes the minimum agent
greater than $\Agent$ for which
$\Weight{\alpha_{1,\Agent'}}=\Weight{\beta_{1,\Agent'}}$. (The agent
$\Agent'$ is well-defined since $\Weight{\alpha}=\Weight{\beta}$.)

\begin{ob}\label{ob:pivot}
	Let $(\alpha,\beta)$ belong to $\Encs{n}{n}$, and let
	$(\alpha',\beta)$ be the result of applying a pivot operation to
	$(\alpha,\beta)$ at agent $\Agent$.  Then $(\alpha',\beta)$ belongs to
	$\Encs{n}{n}$ and $\Span{[\alpha',\beta]}{\Agent'}$ is contained in
	$\Span{[\alpha,\beta]}{\Agent'}$ for all agents $\Agent'$ in
	$[n]-\Agent$.
\end{ob}

For any $(\alpha,\beta)$ in $\Encs{n}{n}$, and any object $\Object$ in
$[n-1]$ such that
$\Weight{\alpha_{1,\Object}}>\Weight{\beta_{1,\Object}}$,
$\beta_\Object=0$, and $\beta_{\Object+1}=1$, we say that an unsort
operation is applicable to $(\alpha,\beta)$ across objects $\Object$
and $\Object+1$.  The result of applying this operation is the pair of
binary strings $(\alpha,\beta')$ that is the same as $(\alpha,\beta)$
except $\beta'_\Object=1$ and $\beta'_{\Object+1}=0$.

\begin{ob}\label{ob:unsort}
	Let $(\alpha,\beta)$ belong to $\Encs{n}{n}$, and let
	$(\alpha,\beta')$ be the result of applying an unsort operation to
	$(\alpha,\beta)$ across objects $\Object$ and $\Object+1$. Then
	$(\alpha,\beta')$ belongs to $\Encs{n}{n}$ and
	$\Span{[\alpha,\beta']}{\Agent}$ is contained in
	$\Span{[\alpha,\beta]}{\Agent}$ for all agents $\Agent$ in $[n]$.
\end{ob}

\subsubsection{Structural Properties of Matchings in $\Decs{n}$}

Claim \ref{claim:decodings} below gives a a simple characterization of
matchings in $\DecsBinary{n}$, and hence implies that $\DecsBinary{n} = \Decs{n}$.
\begin{claim}\label{claim:decodings}
 Let $\Matching$ be a matching in $\DecsBinary{n}$ such that $|\Matching|<n$,
 let $\Agent$ denote $|\Matching|$, let $\Agent'$ denote $\Agent+1$,
 let $\Object$ denote $\MaxMatched{\Matching}$, and let $\Matching'$
 denote $\Matching+(\Agent',\Object^*)$. Then $\Matching'$ belongs to
 $\DecsBinary{n}$ if and only if $\Object^*=\MinUnmatched{\Matching}$ or
 $\Object<\Object^*\leq n$.
\end{claim}


\renewcommand{\Decs}[1]{\DecsBinary{#1}}
\begin{proof}
Since $\Matching$ belongs to $\Decs{n}$, there exists $(\alpha,\beta)$
in $\Encs{\Agent}{\Object}$ such that $\Matching=[\alpha,\beta]$.
Let $\Object'$ denote $\MaxMatched{\Matching'}$.

For the ``if'' direction, we need to prove that there exists
$(\alpha',\beta')$ in $\Encs{\Agent'}{\Object'}$ such that
$\Matching'=[\alpha',\beta']$.  We consider two cases.

Case~1: $\Object^*=\MinUnmatched{\Matching}\leq\Object$.
Observation~\ref{ob:extend} implies that $(\alpha 0,\beta)$ belongs to
$\Encs{\Agent'}{\Object'}$ and $\Matching'=[\alpha 0,\beta]$.
	
Case~2: $\Object<\Object^*\leq n$. Let $k$ denote
$\Object^*-\Object-1$. Observation~\ref{ob:extend} implies that
$(\alpha 1,\beta 0^k1)$ belongs to $\Encs{\Agent'}{\Object'}$ and
$\Matching'=[\alpha 1,\beta 0^k1]$.

We now address the ``only if'' direction. Assume that
$\Matching'$ belongs to $\Decs{n}$.  Since
$\Matching'=\Matching+(\Agent',\Object^*)$, we deduce that $\Object^*$
belongs to $[n]\setminus\ObjectsOf{\Matching}$. It remains to prove that
$\Object^*=\MinUnmatched{\Matching}$ or $\Object<\Object^*$.  Let
$\Objects$ denote the set of objects
$[\Object]\setminus\ObjectsOf{\Matching}$. We consider two cases.

Case~1: $\Agent=\Object$. Thus $\Objects=\emptyset$ and since
$\Object^*$ is unmatched in $\Matching$, we have $\Object<\Object^*$.

Case~2: $\Agent<\Object$. Since $\Matching'$ belongs to $\Decs{n}$,
there exists $(\alpha',\beta')$ in $\Encs{\Agent'}{\Object'}$ such
that $\Matching'=[\alpha',\beta']$. Observation~\ref{ob:prefix}
implies that $(\alpha'_{1,\Agent},\beta'_{1,\Object})$ belongs to
$\Encs{\Agent}{\Object}$ and
$(\alpha'_{1,\Agent},\beta'_{1,\Object})\Congruent(\alpha,\beta)$.
Since\\
$(\alpha'_{1,\Agent},\beta'_{1,\Object})\Congruent(\alpha,\beta)$,
Observation~\ref{ob:congruent} implies that
$[\alpha'_{1,\Agent},\beta'_{1,\Object}]=[\alpha,\beta]=\Matching$.
Let $\Objects_0$ denote the set of all objects $i$ in $\Objects$ such
that $\beta'_i=0$. Observation~\ref{ob:encode} implies that
$\Objects_0=\Objects$. We consider two cases.

Case~2.1: $\alpha'_{\Agent'}=1$. Since
$\Matching'(\Agent')=\Object^*$, Observation~\ref{ob:encode} implies
that $\beta'_{\Object^*}=1$. Thus $\Object^*$ does not belong to
$\Objects=\Objects_0$. Since $\Object^*$ is unmatched in $\Matching$,
we conclude that $\Object<\Object^*$.

Case~2.2: $\alpha'_{\Agent'}=0$. Since $\Agent<\Object$ and
$(\alpha'_{1,\Agent},\beta'_{1,\Object})$ belongs to
$\Encs{\Agent}{\Object}$, we have $\beta'_\Object=1$ and 
$\Weight{\alpha'_{1,\Agent}}=\Weight{\beta'_{1,\Object}}>\Weight{\beta'_{1,\Agent}}$.
Let $k$ denote the number of $0$'s in $\alpha'_{1,\Agent}$, and let
$\ell$ denote the number of $0$'s in $\beta'_{1,\Agent}$.  Since
$\Weight{\alpha'_{1,\Agent}}>\Weight{\beta'_{1,\Agent}}$, we have
$\ell>k$.  Let $\Objects'$ denote the indices of the first $k$ $0$'s
in $\beta'_{1,\Agent}$, and let $\Objects''$ denote the indices of the
remaining $\ell-k$ $0$'s in $\beta'_{1,\Agent}$.  Since
$\Matching=[\alpha'_{1,\Agent},\beta'_{1,\Object}]$, we deduce that
the objects in $\Objects'$ are all matched in $\Matching$ and the
objects in $\Objects''$ are all unmatched in $\Matching$. Thus
$\Objects\cap[\Agent]=\Objects_0\cap[\Agent]=\Objects''\not=\emptyset$,
It follows that $\MinUnmatched{\Matching}$ is the minimum object in
$\Objects''$.  Since $\Matching'=[\alpha',\beta']$,
$\alpha'_{\Agent'}=0$, and there are $k+1$ $0$'s in $\alpha'$, we deduce
that $\Object^*$ is the minimum object in $\Objects''$. Thus
$\Object^*=\MinUnmatched{\Matching}$.
\end{proof}

\renewcommand{\Decs}[1]{\MatchingsSymbol^*(#1)}

Lemma~\ref{lem:reach} below establishes a one-to-one
correspondence between matchings that are reachable from $\Config_\Oaf$
and matchings in $\Decs{n}$. 
\begin{restatable}{lemma}{reach}\label{lem:reach}
Let $\Oaf=(n,\Prefs)$ be an OAF. Then $\Matching$ belongs to
$\Reach{\Config_\Oaf}$ if and only if $\Matching$ belongs to
$\Decs{n}$, $|\Matching|=n$, and $\IndRat{\Prefs}{\Matching}$ holds.
\end{restatable}
\punt{
\begin{proof}
First we prove the ``only if'' direction. Suppose that $\Matching$
belongs to $\Reach{\Config_\Oaf}$. Then there is a sequence
$\Matching_\Oaf=\Matching_0,\ldots,\Matching_k=\Matching$ of perfect
matchings of $\Oaf$ such that
$\Swap{\Oaf}{\Matching_{i-1}}{\Matching_i}$ for all $i$ in $[k]$.

For any $i$ such that $0\leq i\leq k$, let $P(i)$ denote the predicate
asserting that the following conditions hold: $\Matching_i$ belongs to
$\Decs{n}$; $|\Matching_i|=n$; $\IndRat{\Prefs}{\Matching_i}$
holds. We prove by induction on $i$ that $P(i)$ holds for all $i$ in
$\{0,\ldots,k\}$.  Using the definition of $\Matching_\Oaf$, it is
easy to see that $P(0)$ holds. Now consider the induction step. Fix
$i$ in $[k]$ and assume that $P(i-1)$ holds.  We need to prove that
$P(i)$ holds. Since $P(i-1)$ holds, we know that $\Matching_{i-1}$
belongs to $\Decs{n}$, $|\Matching_{i-1}|=n$, and
$\IndRat{\Prefs}{\Matching_{i-1}}$.  Since $\Matching_{i-1}$ belongs
to $\Decs{n}$ and $|\Matching_{i-1}|=n$, there exists $(\alpha,\beta)$
in $\Encs{n}{n}$ such that $\Matching_{i-1}=[\alpha,\beta]$. Let
$\Object$ denote the object in $[n-1]$ such that
$\Swap{\Oaf,(\Object,\Object+1)}{\Matching_{i-1}}{\Matching_i}$.

Let $(\alpha',\beta')$ and $(\alpha'',\beta'')$ be defined as
follows. First, if a complement operation is applicable to
$(\alpha,\beta)$ at agent $\Object$ and $\beta_{\Object}=0$, then
$(\alpha',\beta')$ is the result of applying this operation to
$(\alpha,\beta)$, and otherwise $(\alpha',\beta')$ is equal to
$(\alpha,\beta)$.  Second, if a complement operation is applicable to
$(\alpha,\beta)$ at agent $\Object+1$ and $\beta_{\Object+1}=1$, then
$(\alpha'',\beta'')$ is the result of applying this operation to
$(\alpha',\beta')$, and otherwise $(\alpha'',\beta'')$ is equal to
$(\alpha',\beta')$. Observation~\ref{ob:congruent} implies that
$(\alpha'',\beta'')$ belongs to $\Encs{n}{n}$ and
$\Matching_{i-1}=[\alpha'',\beta'']$.

Using Observation~\ref{ob:encode}, it is straightforward to prove that
$\beta''_\Object=1$ and $\beta''_{\Object+1}=0$. Thus a sort operation
is applicable to $(\alpha'',\beta'')$ across objects $\Object$ and
$\Object+1$; let $(\alpha'',\beta''')$ denote the result of applying
this sort operation.  Observation~\ref{ob:sort} implies that
$(\alpha'',\beta''')$ belongs to $\Encs{n}{n}$ and
$\Matching_i=[\alpha'',\beta''']$. Thus $\Matching_i$ belongs to
$\Decs{n}$ and $|\Matching_i|=n$.  Since $\IndRat{\Prefs}{\Matching_{i-1}}$
holds and $\Swap{\Oaf}{\Matching_{i-1}}{\Matching_i}$, we deduce that
$\IndRat{\Prefs}{\Matching_i}$ holds. We conclude that $P(i)$ holds,
completing the proof by induction.

We now prove the ``if'' direction.  Assume that $\Matching$ belongs to
$\Decs{n}$, $|\Matching|=n$, and $\IndRat{\Prefs}{\Matching}$
holds. Observation~\ref{ob:encodings} implies there exists
$(\alpha,\beta)$ in $\Encs{n}{n}$ such that
$\Steps{(0^n,0^n)}{(\alpha,\beta)}$. It follows that there is a
sequence
$(0^n,0^n)=(\alpha^{(0)},\beta^{(0)}),\ldots,(\alpha^{(k)},\beta^{(k)})=(\alpha,\beta)$
of pairs in $\Encs{n}{n}$ such that an applicable complement or sort
operation transforms $(\alpha^{(i-1)},\beta^{(i-1)})$ into
$(\alpha^{(i)},\beta^{(i)})$ for all $i$ in $[k]$.  Let $\Matching_i$
denote $[\alpha^{(i)},\beta^{(i)}]$ for all $i$ such that $0\leq i\leq
k$. Thus $\Matching_0=\Matching_\Oaf$ and $\Matching_k=\Matching$.

Observations~\ref{ob:congruent} and~\ref{ob:sort} imply that for all
$i$ in $[k]$, either $\Matching_i=\Matching_{i-1}$ or $\Matching_i$ is
obtained from $\Matching_{i-1}$ via an exchange across two adjacent
objects. It remains to prove that any such exchanges are swaps, i.e.,
do not violate individual rationality. Below we accomplish this by
proving that $\IndRat{\Prefs}{\Matching_i}$ holds for $0\leq i\leq k$.

For any $i$ in $[k]$, let $P(i)$ denote the predicate
$\Span{\Matching_{i-1}}{\Agent}\subseteq\Span{\Matching_i}{\Agent}$
for all agents $\Agent$ in $[n]$.  We claim that $P(i)$ holds for all
$i$ in $[k]$. To prove the claim, fix an integer $i$ in $[k]$.  We
consider two cases.

Case~1: A complement operation transforms
$(\alpha^{(i-1)},\beta^{(i-1)})$ into
$(\alpha^{(i)},\beta^{(i)})$.  In this case,
Observation~\ref{ob:congruent} implies that
$\Matching_i=\Matching_{i-1}$. Thus
$\Span{\Matching_{i-1}}{\Agent}=\Span{\Matching_i}{\Agent}$ for all
agents $\Agent$ in $[n]$.

Case~2: A sort operation transforms $(\alpha^{(i-1)},\beta^{(i-1)})$
into $(\alpha^{(i)},\beta^{(i)})$. Assume that the sort operation
is applied to $(\alpha^{(i-1)},\beta^{(i-1)})$ across objects
$\Object$ and $\Object+1$. Then $\beta_\Object^{(i-1)}=1$ and
$\beta_{\Object+1}^{(i)}=0$, and Observation~\ref{ob:encode} implies
that $\Matching_{i-1}^{-1}(\Object)\leq\Object$ and
$\Matching_{i-1}^{-1}(\Object+1)\geq\Object+1$. Thus
Observation~\ref{ob:sort} implies
$\Span{\Matching_{i-1}}{\Agent}\subseteq\Span{\Matching_i}{\Agent}$
for all agents $\Agent$ in $[n]$.

Since $P(i)$ holds for all $i$ in $[k]$ and
$\IndRat{\Prefs}{\Matching_k}$ holds, we deduce that
$\IndRat{\Prefs}{\Matching_i}$ holds for all $i$ such that $0\leq
i\leq k$, as required.
\end{proof}

}

The next three lemmas are concerned with enlarging a given matching
$\Matching$ in $\Decs{n}$ such that $|\Matching|<n$ by introducing a
suitable match for agent $|\Matching|+1$. Lemma~\ref{lem:left}
(resp., Lemma~\ref{lem:right}) addresses the case where agent
$|\Matching|+1$ is matched to an object that is at most (resp., at
least) $|\Matching|+1$. By combining these two lemmas, we obtain Lemma~\ref{lem:options}, which shows that it is
sufficient to consider matching agent $|\Matching|+1$ with an object
in $\{\MinUnmatched{\Matching},\MaxMatched{\Matching}+1\}$.

\begin{restatable}{lemma}{left}\label{lem:left}
Let $\Oaf=(n,\Prefs)$ be an OAF, let $\Matching$ be a matching in
$\Decs{n}$ such that $|\Matching|<n$, let $\Agent$ denote
$|\Matching|$, let $\Agent'$ denote $\Agent+1$, let $\Object$ denote
$\MaxMatched{\Matching}$, and let $\Matching'$ denote
$\Matching+(\Agent',\MinUnmatched{\Matching})$. Assume that
$\Agent<\Object$,
$\ReachAgain{\Config_\Oaf}{\Matching}\not=\emptyset$, and
$\Left{\Prefs}{\Agent'}\leq\MinUnmatched{\Matching}$. Then
$\ReachAgain{\Config_\Oaf}{\Matching'}\not=\emptyset$.
\end{restatable}
\punt{
\begin{proof}
Let $\Matching^*$ be a matching in
$\ReachAgain{\Config_\Oaf}{\Matching}$. Lemma~\ref{lem:reach} implies
that $\IndRat{\Prefs}{\Matching^*}$ holds and there exists
$(\alpha^*,\beta^*)$ in $\Encs{n}{n}$ such that
$\Matching^*=[\alpha^*,\beta^*]$.  Using
Observations~\ref{ob:congruent} and~\ref{ob:prefix}, we deduce that
$[\alpha^*_{1,\Agent},\beta^*_{1,\Object}]$ is equal to $\Matching$.
We consider two cases.

Case~1: $\alpha^*_{\Agent'}=0$.  Using Observation~\ref{ob:extend}, we
deduce that $[\alpha^*_{1,\Agent'},\beta^*_{1,\Object}]$ is equal to
$\Matching'$.  Hence Observation~\ref{ob:prefix} implies that
$\Matching^*$ contains $\Matching'$.  Thus
$\ReachAgain{\Config_\Oaf}{\Matching'}\not=\emptyset$, as required.
	
Case~2: $\alpha^*_{\Agent'}=1$. In this case, it is straightforward to
prove that a pivot operation is applicable to $(\alpha^*,\beta^*)$ at
agent $\Agent'$; let $(\alpha^{**},\beta^*)$ denote the result of this
operation.  Observation~\ref{ob:pivot} implies that
$(\alpha^{**},\beta^*)$ belongs to $\Encs{n}{n}$. Let $\Matching^{**}$
denote $[\alpha^{**},\beta^*]$.  Thus $\Matching^{**}$ belongs to
$\Decs{n}$.  Observation~\ref{ob:pivot} implies that
$\IndRatAgent{\Prefs}{\Matching^{**}}{\Agent''}$ holds for all agents
$\Agent''$ in $[n]-\Agent'$.  Since the inequality
$\Left{\Prefs}{\Agent'}\leq\MinUnmatched{\Matching}$ implies that
$\IndRatAgent{\Prefs}{\Matching^{**}}{\Agent'}$ holds, we deduce that
$\IndRat{\Prefs}{\Matching^{**}}$ holds. Thus Lemma~\ref{lem:reach}
implies that $\Matching^{**}$ belongs to $\Reach{\Config_\Oaf}$.
Since $\alpha^{**}_{1,\Agent}=\alpha^*_{1,\Agent}$,
Observation~\ref{ob:prefix} implies that $\Matching^{**}$ contains
$\Matching$. Using Observation~\ref{ob:extend}, we deduce that
$[\alpha^{**}_{1,\Agent'},\beta^*_{1,\Object}]$ is equal to
$\Matching'$. Hence Observation~\ref{ob:prefix} implies that
$\Matching^{**}$ contains $\Matching'$.  Since $\Matching^{**}$
belongs to $\Reach{\Config_\Oaf}$ and $\Matching^{**}$ contains
$\Matching'$, we conclude that $\Matching^{**}$ is contained in
$\ReachAgain{\Config_\Oaf}{\Matching'}$.
\end{proof}

}

\begin{restatable}{lemma}{right}\label{lem:right}
Let $\Oaf=(n,\Prefs)$ be an OAF, let $\Matching$ be a matching in
$\Decs{n}$ such that $|\Matching|<n$, let $\Agent$ denote
$|\Matching|$, let $\Agent'$ denote $\Agent+1$, let $\Object$ denote
$\MaxMatched{\Matching}$, let $\Matching'$ denote
$\Matching+(\Agent',\Object')$ where $\Object'$ belongs to
$[n]\setminus[\Object+1]$ and
$\ReachAgain{\Config_\Oaf}{\Matching'}\not=\emptyset$, and let
$\Matching''$ denote $\Matching+(\Agent',\Object+1)$. Then
$\ReachAgain{\Config_\Oaf}{\Matching''}\not=\emptyset$.
\end{restatable}
\punt{
\begin{proof}
Let $\Matching^*$ be a matching in
$\ReachAgain{\Config_\Oaf}{\Matching'}$. Lemma~\ref{lem:reach} implies
that $\IndRat{\Prefs}{\Matching^*}$ holds and there exists
$(\alpha^*,\beta^*)$ in $\Encs{n}{n}$ such that
$\Matching^*=[\alpha^*,\beta^*]$. Using Observations~\ref{ob:encode}
and~\ref{ob:prefix}, we find that
$\alpha^*_{\Agent'}=\beta^*_{\Object'}=1$ and $\beta^*_{\Object+k}=0$
for all $k$ in $[\Object'-\Object-1]$.  Using
Observation~\ref{ob:unsort}, we deduce that a sequence of
$\Object'-\Object-1$ unsort operations can be used to transform
$(\alpha^*,\beta^*)$ into a pair of binary strings
$(\alpha^*,\beta^{**})$ in $\Encs{n}{n}$ such that
$\beta^{**}_{1,\Object}=\beta^*_{1,\Object}$ and
$\IndRat{\Prefs}{[\alpha^*,\beta^{**}]}$ holds. Let $\Matching^{**}$
denote $[\alpha^*,\beta^{**}]$; thus $\Matching^{**}$ belongs to
$\Decs{n}$. Lemma~\ref{lem:reach} implies that $\Matching^{**}$
belongs to $\Reach{\Config_\Oaf}$.  Since
$\beta^{**}_{1,\Object}=\beta^*_{1,\Object}$,
Observation~\ref{ob:prefix} implies that $\Matching^{**}$ contains
$\Matching$. Using Observation~\ref{ob:extend}, we deduce that
$[\alpha^*_{1,\Agent'},\beta^{**}_{1,\Object+1}]$ is equal to
$\Matching''$. Hence Observation~\ref{ob:prefix} implies that
$\Matching^{**}$ contains $\Matching''$.  Since $\Matching^{**}$
belongs to $\Reach{\Config_\Oaf}$ and $\Matching^{**}$ contains
$\Matching''$, we conclude that $\Matching^{**}$ belongs to
$\ReachAgain{\Config_\Oaf}{\Matching''}$.
\end{proof}

}

\begin{restatable}{lemma}{options}\label{lem:options}
Let $\Oaf=(n,\Prefs)$ be an OAF, let $\Matching$ be a matching in
$\Decs{n}$ such that $|\Matching|<n$, let $\Agent$ denote
$|\Matching|$, let $\Agent'$ denote $\Agent+1$, let $\Object$ denote
$\MaxMatched{\Matching}$, and assume that
$\ReachAgain{\Config_\Oaf}{\Matching}\not=\emptyset$. Then there
exists a matching in $\ReachAgain{\Config_\Oaf}{\Matching}$ that
matches agent $\Agent'$ to an object in
$\{\MinUnmatched{\Matching},\Object+1\}$.
\end{restatable}
\punt{
\begin{proof}
Let $\Matching^*$ belong to $\ReachAgain{\Config_\Oaf}{\Matching}$ and
let $\Object^*$ denote $\Matching^*(\Agent')$.  Lemma~\ref{lem:reach}
implies that $\Matching^*$ belongs to $\Decs{n}$ and
$\IndRat{\Prefs}{\Matching^*}$ holds.

Since $\Matching^*$ belongs to $\Decs{n}$, the definition of $\Decs{n}$
implies that $\Object^*=\MinUnmatched{\Matching}$ or
$\Object<\Object^*$.  We consider two cases.
	
Case~1: $\Object^*=\MinUnmatched{\Matching}$. Since
$\IndRat{\Prefs}{\Matching^*}$ holds, we deduce that
$\Left{\Prefs}{\Agent'}\leq\Object^*$. Hence the claim of the lemma
follows from Lemma~\ref{lem:left}.
	
Case~2: $\Object<\Object^*$.  Since $\IndRat{\Prefs}{\Matching^*}$
holds, we deduce that $\Object^*\leq\Right{\Prefs}{\Agent'}$.  If
$\Object^*=\Object+1$, the claim of the lemma is immediate.
Otherwise, it follows from Lemma~\ref{lem:right}.
\end{proof}

}


We are now ready to prove the main technical lemma of this section, Lemma~\ref{lem:main} below.

%

\begin{lemma}\label{lem:main}
Consider an execution of Algorithm~\ref{alg:greedyPath} with inputs
$\Oaf=(n,\Prefs)$, $\Matching_0$, and $\Matching_1$. If the guard of
the while loop is evaluated in a state where
$\Matching\not=\emptyset$, then the following conditions hold in that
state: (1) $\Matching$ belongs to $\Decs{n}$; (2)
$\IndRat{\Prefs}{\Matching}$ holds; (3)
$\ReachAgain{\Config_\Oaf}{\Matching_1}\not=\emptyset$ implies
$\ReachAgain{\Config_\Oaf}{\Matching}\not=\emptyset$.
Furthermore, if the guard of the while loop is evaluated in a state
where $\Matching=\emptyset$, then
$\ReachAgain{\Config_\Oaf}{\Matching_1}=\emptyset$.
\end{lemma}

\begin{proof}
We prove the claim by induction on the number of iterations of the
loop. For the base case, we verify that the stated conditions hold
the first time the loop is reached. The initialization of $\Matching$
ensures that $\Matching\not=\emptyset$, so we need to verify
conditions~(1) through~(3). Since
$\ReachAgain{\Config_\Oaf}{\Matching_0}\not=\emptyset$,
Lemma~\ref{lem:reach} implies that $\IndRat{\Prefs}{\Matching_0}$
holds, and Lemma \ref{lem:reach} and the definition of $\Decs{n}$
together 
imply that $\Matching_0$ belongs to $\Decs{n}$. Since
$\Object_0\leq\Right{\Prefs}{|\Matching_1|}\leq n$,
the definition of $\Decs{n}$ 
implies that $\Matching_1$ belongs to
$\Decs{n}$. Since $\IndRat{\Prefs}{\Matching_0}$ holds and the
preconditions associated with Algorithm~\ref{alg:greedyPath} ensure
that $\IndRatAgent{\Prefs}{\Matching_1}{|\Matching_1|}$ holds, we
deduce that $\IndRat{\Prefs}{\Matching_1}$ holds. Since $\Matching$
is initialized to $\Matching_1$, we conclude that conditions~(1)
through~(3) hold the first time the loop is reached.
 
For the induction step, consider an arbitrary iteration of the loop.
Such an iteration begins in a state where the guard of the while loop
evaluates to true, so we can assume that $0<|\Matching|<n$ and that
conditions~(1) through~(3) hold in this state. Let $\Agent$ denote
$|\Matching|$, let $\Agent'$ denote $\Agent+1$, let $\Object$ denote
$\MaxMatched{\Matching}$, and let $\Matching'$ denote the value of the
program variable $\Matching$ immediately after this iteration of the
loop body. We consider two cases.
 
Case~1: $\Matching'=\emptyset$. In this case, we need to prove that
$\ReachAgain{\Config_\Oaf}{\Matching_1}=\emptyset$. Assume for
the sake of contradiction that
$\ReachAgain{\Config_\Oaf}{\Matching_1}\not=\emptyset$.
Condition~(3) implies that
$\ReachAgain{\Config_\Oaf}{\Matching}\not=\emptyset$. Thus
Lemma~\ref{lem:options} implies there exists a matching
$\Matching^*$ in $\ReachAgain{\Config_\Oaf}{\Matching}$ such
that $\Left{\Prefs}{\Agent'}\leq\MinUnmatched{\Matching}$ or
$\Object+1\leq\Right{\Prefs}{\Agent'}$. It follows by inspection of
the code that $\Matching'\not=\emptyset$, a contradiction.
 
Case~2: $\Matching'\not=\emptyset$. In this case, we need to prove
that conditions~(1) through~(3) hold with $\Matching$ replaced by
$\Matching'$; we refer to these conditions as postconditions~(1)
through~(3). Since
$\Matching'(\Agent')\leq\Right{\Prefs}{\Agent'}\leq n$ and
condition~(1) implies that $\Matching$ belongs to $\Decs{n}$,
the definition of $\Decs{n}$
implies that $\Matching'$ belongs to
$\Decs{n}$. Thus postcondition~(1) holds. Since condition~(2) implies
that $\IndRat{\Prefs}{\Matching}$ holds, and
$\IndRatAgent{\Prefs}{\Matching'}{\Agent'}$ holds by inspection of the
code, we deduce that postcondition~(2) holds. It remains to establish
postcondition~(3). In order to do so, we may assume that
$\ReachAgain{\Config_\Oaf}{\Matching_1}\not=\emptyset$.
Condition~(3) implies that
$\ReachAgain{\Config_\Oaf}{\Matching}\not=\emptyset$. To
establish postcondition~(3), we need to prove that
$\ReachAgain{\Config_\Oaf}{\Matching'}\not=\emptyset$. We
consider two cases.
 
Case~2.1: $\Agent<\Object$ and
$\Left{\Prefs}{\Agent'}\leq\MinUnmatched{\Matching}$. In this case,
$\Matching'=\Matching+(\Agent',\MinUnmatched{\Matching})$, and
Lemma~\ref{lem:left} implies that
$\ReachAgain{\Config_\Oaf}{\Matching'}\not=\emptyset$.
 
Case~2.2: $\Agent=\Object$ or
$\MinUnmatched{\Matching}<\Left{\Prefs}{\Agent'}$. In this case,
$\Matching'=\Matching+(\Agent',\Object+1)$, and
Lemma~\ref{lem:options} implies there exists a matching $\Matching^*$
in $\ReachAgain{\Config_\Oaf}{\Matching}$ such that
$\Matching^*(\Agent')$ is either $\MinUnmatched{\Matching}$ or
$\Object+1$. Since $\IndRat{\Prefs}{\Matching^*}$ holds, the Case~2.2
condition implies that if
$\Matching^*(\Agent')=\MinUnmatched{\Matching}$, then
$\Agent=\Object$, in which case $\ObjectsOf{\Matching}=[\Object]$ and
hence $\MinUnmatched{\Matching}=\Object+1$. It follows that
$\Matching^*(\Agent')=\Object+1$, and hence that
$\ReachAgain{\Config_\Oaf}{\Matching'}\not=\emptyset$.
\end{proof}

Using Lemma~\ref{lem:main}, it is straightforward to establish the correctness of Algorithm~\ref{alg:greedyPath}. 

\begin{restatable}{lemma}{correctness}\label{lem:correctness}
Consider an execution of Algorithm~\ref{alg:greedyPath} with inputs
$\Oaf=(n,\Prefs)$, $\Matching_0$, and $\Matching_1$. The execution
terminates correctly within $n-|\Matching_1|$ iterations.
\end{restatable}
\punt{
\begin{proof}
Lemma~\ref{lem:main} implies that each iteration of
Algorithm~\ref{alg:greedyPath} either increments the cardinality of
matching $\Matching$ or reduces it to zero.  In the latter case, the
algorithm terminates immediately. It follows that the algorithm
terminates within $n-|\Matching_1|$ iterations.  Next, we argue that
the algorithm terminates correctly. In what follows, let $\Matching^*$
denote the final value of the program variable $\Matching$. We
consider two cases.
	
Case~1: Algorithm~\ref{alg:greedyPath} terminates with
$\Matching^*=\emptyset$. In this case, Lemma~\ref{lem:main} implies
that $\ReachAgain{\Config_\Oaf}{\Matching_1}=\emptyset$, as
required.
	
Case~2: Algorithm~\ref{alg:greedyPath} terminates with
$\Matching^*\not=\emptyset$. Lemma~\ref{lem:main} implies that
conditions~(1) through~(3) in the statement of Lemma~\ref{lem:main}
hold with $\Matching$ replaced by $\Matching^*$; we refer to these
conditions as postconditions~(1) through~(3). Since
$\Matching^*\not=\emptyset$, no agent-object pairs are removed from
$\Matching$ during the execution of Algorithm~\ref{alg:greedyPath}.
Since Algorithm~\ref{alg:greedyPath} initializes $\Matching$ to
$\Matching_1$, we deduce that $\Matching_1$ is contained in
$\Matching^*$.  Since the guard of the loop evaluates to false when
$\Matching$ is equal to $\Matching^*$ and postcondition~(1) holds, we
deduce that $\Matching^*$ belongs to $\Decs{n}$ and
$|\Matching^*|=n$. Since postconditions~(1) and~(2) hold,
Lemma~\ref{lem:reach} implies that $\Matching^*$ belongs to
$\Reach{\Config_\Oaf}$.  Since $\Matching_1$ is contained in
$\Matching^*$, we deduce that $\Matching^*$ belongs to
$\ReachAgain{\Config_\Oaf}{\Matching_1}$, as required.
\end{proof}

}


\subsection{Object Reachability}
\label{sec:pathReachObject}

We now describe how to use Algorithm~\ref{alg:greedyPath} to solve the
reachable object problem on paths in $O(n^2)$ time. Let
$\Oaf=(n,\Prefs)$ be a given OAF, let $\Agent^*$ be an agent in $[n]$,
and let $\Object^*$ be an object in $[n]$. Assume without loss of
generality that $\Agent^*<\Object^*$. We wish to determine whether
there is a matching in $\Reach{\Config_\Oaf}$ that matches
$\Agent^*$ to $\Object^*$, and if so, to compute such a matching. We
begin by using a preprocessing phase to compute
$\Left{\Prefs}{\Agent}$ and $\Right{\Prefs}{\Agent}$ for all agents
$\Agent$ in $[n]$. We start with agent $\Agent^*$, and check whether
$\Left{\Prefs}{\Agent^*}\leq\Object^* \le \Right{\Prefs}{\Agent^*}$. If
this check fails, we halt and report failure. Otherwise, we proceed to
the remaining agents. Barring failure, the overall cost of the
preprocessing phase is $O(n^2)$. We now describe how to proceed in
the special case where $\Agent^*$ is the leftmost agent on the path,
i.e., where $\Agent^*=1$. Later we will see how to efficiently reduce
the general case to this special case. In the special case
$\Agent^*=1$, we call Algorithm~\ref{alg:greedyPath} with
$\Matching_0=\emptyset$ and $\Matching_1=\{(1,\Object^*)\}$. If there is
a matching in $\Reach{\Config_\Oaf}$ that matches agent $1$ to
object $\Object$, then Algorithm~\ref{alg:greedyPath} returns such a
matching. If not, Algorithm~\ref{alg:greedyPath} returns the empty
matching. Excluding the cost of the preprocessing phase, the time
complexity of Algorithm~\ref{alg:greedyPath} is $O(n)$. So the overall
running time is $O(n^2)$, as it is dominated by the preprocessing
phase. We now discuss how to reduce the case of general $\Agent^*$ to
the special case $\Agent^*=1$. The key is Lemma~\ref{lem:truncate}
below. Informally, Lemma~\ref{lem:truncate} tells us that we can
ignore all of the agents and objects in $[\Agent^*-1]$. Doing this,
$\Agent^*$ is once again leftmost, and we can proceed as in the
special case $\Agent^*=1$. Alternatively, we can run
Algorithm~\ref{alg:greedyPath} with $\Matching_0=\{(i,i)\mid i\in
[\Agent^*-1]\}$ and
$\Matching_1=\Matching_0+\{(\Agent^*,\Object^*)\}$. 
We now proceed to prove Lemma~\ref{lem:truncate}. 

\begin{restatable}{lemma}{truncate}\label{lem:truncate}
Let $\Oaf=(n,\Prefs)$ be an OAF, let $\Matching$ be a matching in
$\Reach{\Config_\Oaf}$, let $\Agent$ belong to
$\AgentsOf{\Matching}$, let $\Object$ denote $\Matching(\Agent)$, and
assume that $\Agent<\Object$. Then there is a matching $\Matching'$
in $\Reach{\Config_\Oaf}$ such that
$\Matching'(\Agent)=\Object$ and $\Matching'(\Agent')=\Agent'$ for all
agents $\Agent'$ in $[\Agent-1]$.
\end{restatable}
\punt{

For any $(\alpha,\beta)$ in $\Encs{n}{n}$, we say that a cancel
operation is applicable to $(\alpha,\beta)$ if $\Weight{\alpha}>0$.
The result of applying this operation is the pair of binary strings
$(\alpha',\beta')$ that is the same as $(\alpha,\beta)$ except that
the first appearing $1$ in $\alpha$ (resp., $\beta$) is changed to a
$0$ in $\alpha'$ (resp., $\beta'$). It is easy to see that
$(\alpha',\beta')$ belongs to $\Encs{n}{n}$.
\begin{ob}\label{ob:cancel}
	Let $(\alpha,\beta)$ belong to $\Encs{n}{n}$, let $\Agent$ and
	$\Agent'$ be agents in $[n]$ such that
	$\alpha_{\Agent}=\alpha_{\Agent'}=1$ and $\Agent>\Agent'$, let
	$(\alpha',\beta')$ be the result of applying a cancel operation to
	$(\alpha,\beta)$, let $\Matching$ denote $[\alpha,\beta]$, and let
	$\Matching'$ denote $[\alpha',\beta']$, Then
	$\Matching'(\Agent)=\Matching(\Agent)$ and
	$\Span{\Matching'}{\Agent''}$ is contained in
	$\Span{\Matching}{\Agent''}$ for all agents $\Agent''$ in $[n]$.
\end{ob}

\begin{proof}
Lemma~\ref{lem:reach} implies that $\IndRat{\Prefs}{\Matching}$ holds
and there exists $(\alpha,\beta)$ in $\Encs{n}{n}$ such that
$\Matching=[\alpha,\beta]$. Since $\Agent<\Object$,
Observation~\ref{ob:encode} implies that
$\alpha_\Agent=\beta_\Object=1$.  Let the number of 1 bits to the left
of position $\Agent$ in $\alpha$ be equal to $k$. Thus the number of 1
bits to the left of position $\Object$ in $\beta$ is also $k$. Let
$(\alpha',\beta')$ be the pair of binary strings in $\Encs{n}{n}$ that
results from applying $k$ cancel operations to $(\alpha,\beta)$. Using
Observation~\ref{ob:cancel}, it is straightforward to check that the
matching $\Matching'=[\alpha',\beta']$ satisfies the requirements of
the lemma.
\end{proof}

}

Theorem~\ref{thm:RO_path} below summarizes the main result of this
section.

\begin{theorem}\label{thm:RO_path}
Reachable object on paths can be solved in $O(n^2)$ time.
\end{theorem}

\subsection{Pareto-Efficient Reachability}
\label{sec:pathPareto}

Let $\Oaf=(n,\Prefs)$ be an OAF, and let $\Config=(\Oaf,\Matching_0)$
be a configuration for which we wish to compute a Pareto-efficient
matching. Below we describe a simple way to use
Algorithm~\ref{alg:greedyPath} to solve Pareto-efficient matching on
paths in $O(n^3)$ time. We then explain how to improve the time bound
to $O(n^2\log n)$, and then to $O(n^2)$. In all cases we employ the
same high-level strategy based on serial dictatorship. We begin by
performing the $O(n^2)$-time preprocessing phase discussed in
Section~\ref{sec:pathReachObject}; we only need to perform this
computation once. After the preprocessing phase, the output matching
is computed in $n$ stages numbered from $1$ to $n$. In stage $k$, we
determine the best possible match that we can provide to agent $k$
while continuing to maintain the previously-determined matches for
agents $1$ through $k-1$. We now describe how to use
Algorithm~\ref{alg:greedyPath} to implement any given stage $k$ in
$O(n^2)$ time. In stage $k$, we call Algorithm~\ref{alg:greedyPath}
$O(n)$ times. In each of these calls, the input matching $\Matching_0$
contains the $k-1$ previously-determined agent-object pairs involving
the agents in $[k-1]$. The calls to Algorithm~\ref{alg:greedyPath}
differ only in terms of the value assigned to the input object
$\Object_0$ which, together with $\Matching_0$, determines
$\Matching_1$. We vary $\Object_0$ over all values meeting the
precondition
$\MaxMatched{\Matching_0}<\Object_0\leq\Right{\Prefs}{|\Matching_1|}$
associated with Algorithm~\ref{alg:greedyPath}; the number of such
values is $O(n)$.

This allows us to determine, in $O(n^2)$ time, the rightmost feasible
match, if any, for agent $k$. By Lemma~\ref{lem:options}, the
leftmost potential match for agent $k$ is
$\MinUnmatched{\Matching_0}$, and this option is only feasible if
$\Left{\Prefs}{k}\leq\MinUnmatched{\Matching_0}$. Since
$\ReachAgain{\Config_\Oaf}{\Matching_0}\not=\emptyset$, we are
guaranteed to find at least one candidate match for agent $k$ in this
process. If there is exactly one candidate, then we select it as the
match of agent $k$. Otherwise, there are two candidates (leftmost and
rightmost), and we select the candidate that agent $k$ prefers.

The simple algorithm described above has a running time of $O(n^2)$
per stage, and hence $O(n^3)$ overall. To understand how to implement
a stage more efficiently, it is useful to assign a color to the
program state each time the condition of the while loop is evaluated.
We color such a state red if $\Matching=\emptyset$. If the state is
red, then the execution is guaranteed to fail (i.e., return
$\emptyset$) immediately. We color such a state green if
$|\Matching|=\MaxMatched{\Matching}>0$. If the state is green, it is
straightforward to prove that the program will proceed to assign every
agent $\Agent$ in $\{|\Matching|+1,\ldots,n\}$ to object $\Agent$, and
then will succeed (i.e., return a nonempty matching). This
observation also implies that if the state is green after a given
number of iterations, it remains green after each subsequent
iteration. If a state is neither red nor green, we color it yellow.

We are now ready to see how to improve the running time of stage $k$
to $O(n\log n)$. As a thought experiment, consider running the $O(n)$
executions of Algorithm~\ref{alg:greedyPath} associated with the
simple algorithm, but now in parallel. Within each of these
executions, we color each successive agent in the set $\{k,\ldots,n\}$
white or black as follows: agent $k$ is colored black; agent
$|\Matching|+1$ is colored white if
$\Left{\Prefs}{|\Matching|+1}\leq\MinUnmatched{\Matching}$, and black
otherwise. The key observation is that as long as none of the parallel
executions have terminated, they all agree on the coloring of the
processed agents. It follows that if we compare two of the parallel
executions, say executions~A and~B where execution~A has a lower value
for $\Object_0$ than execution~B, then execution~A can only transition
to a green state at a strictly earlier iteration than execution~B, and
execution~A cannot transition to a red state earlier than execution~B.
This implies that there is a threshold $\Object_1$ such that all
executions with $\Object_0\leq\Object_1$ succeed, and all of the
remaining executions fail. This in turn means that we do not need to
run all $O(n)$ of the parallel executions of
Algorithm~\ref{alg:greedyPath}. Instead, we can use binary search to
determine the threshold $\Object_1$ in $O(\log n)$ executions. This
observation reduces the running time of a stage from $O(n^2)$ to
$O(n\log n)$.

We now sketch how to further improve the running time of stage $k$ to
$O(n)$. To do so, we will use a single execution of a modified version
of Algorithm~\ref{alg:greedyPath} to compute the threshold $\Object_1$
discussed above. The high-level idea is to treat $\Object_0$ as a
variable instead of a fixed value. Initially, we set $\Object_0$ to
$\MaxMatched{\Matching_0}+1$, the minimum value satisfying the
associated precondition of Algorithm~\ref{alg:greedyPath}. We also
maintain a lower bound $\Object_1'$ on the threshold $\Object_1$. We
initialize $\Object_1'$ to a low dummy value, such as $0$. Each time
the condition of the while loop is evaluated, we check whether the
color of the current state is red, green, or yellow. If the color is
yellow, we continue the execution without altering $\Object_0$ or
$\Object_1'$. If the color is red, then we halt and output the
threshold $\Object_1=\Object_1'$. If the color is green, then we
assign $\Object_1'$ to the current value of $\Object_0$, and we
increment $\Object_0$.

Unfortunately, we cannot simply increment $\Object_0$ and continue the
execution. While incrementing $\Object_0$ has no impact on the
white-black categorization of the agents processed so far, and 
on the matches of the white agents, it causes the match of each
black agent to be incremented. This leads to two difficulties that we
now discuss.

The first difficulty is that there can be a lot of black agents,
making it expensive to maintain an explicit match for each black
agent. Accordingly, when we color an agent black, we do not explicitly
match that agent to a particular object. Instead, we maintain an
ordered list of the black agents. At any given point in the execution,
the black agents are implicitly matched (in the order specified by the
list) to the contiguous block of objects that starts with the current
value of $\Object_0$. Thus, when $\Object_0$ is incremented, the
matches of the black agents are implicitly updated in constant time.

The second difficulty associated with incrementing $\Object_0$ is that
if a black agent $\Agent$ is matched to object
$\Right{\Prefs}{\Agent}$ just before the increment, then it is
infeasible to shift the match of agent $\Agent$ to the right. If this
happens, we need to recognize that executing
Algorithm~\ref{alg:greedyPath} from the beginning with the new higher
value of $\Object_0$ results in a red state, and so we should
terminate. To recognize such events, we introduce an integer variable
called \emph{slack}. We maintain the invariant that \emph{slack} is
equal to the maximum number of positions we can shift the list of
black agents to the right without violating a constraint. We
initialize \emph{slack} to $\Right{\Prefs}{|\Matching_1|}$ minus the
initial value $\MaxMatched{\Matching_0}+1$ of $\Object_0$. When we
color an agent $\Agent$ black, we update \emph{slack} to the minimum
of its current value and $\Right{\Prefs}{\Agent}-\Object_0-\ell$,
where $\ell$ denotes the number of previously-identified black agents.
When we increment $\Object_0$, we decrement \emph{slack} in order to
maintain the invariant. If \emph{slack} becomes negative, we
recognize that the program should be in a red state, and we terminate.

Upon termination, it is straightforward to argue that the output
threshold $\Object_1$ is correct. If $\Object_1$ is equal to the
initial dummy value, then the sole candidate match for agent~$k$ is
object $\MinUnmatched{\Matching_0}$. If not, object $\Object_1$ is a
candidate, and if $\MinUnmatched{\Matching_0}$ is not equal to
$\Object_1$ and $\Left{\Prefs}{k}\leq \MinUnmatched{\Matching_0}$ then
$\MinUnmatched{\Matching_0}$ is a second candidate. If there are two
candidates, we use the preferences of agent~$k$ to select between
them.

Theorem~\ref{thm:PE_path} below summarizes the main result of this
section.

\begin{theorem}
\label{thm:PE_path}
Pareto-efficient matching on paths can be solved in $O(n^2)$ time.
\end{theorem}

	\section{Pareto-Efficient Reachability on Generalized Stars}
\label{sec:gsPeAlg}

Throughout this section, let $\Oaf$ denote an OAF associated with a
generalized star $G$, let $\Center$ denote the center object of $G$,
let $m$ denote the number of branches of $G$, and assume that the
branches are indexed from $1$ to $m$. For any $i$ in $[m]$, let
$\ell_i>0$ denote the number of vertices on branch $i$. We refer to
the objects on branch $i$ as
$\BranchObject{i}{1},\ldots,\BranchObject{i}{\ell_i}$, where object
$\BranchObject{i}{j}$ is at distance $j$ from the center. Let
$\Config_0=(\Oaf,\Matching_0)$ denote the initial configuration, and
let $n=1+\sum_{1\leq i\leq m}\ell_i$ denote the total number of
vertices in $G$.

Our algorithm uses serial dictatorship to compute a Pareto-efficient
matching for configuration $\Config_0$. For any sequence of agents
$\AgentSeq=\Agent_1,\ldots,\Agent_s$ we define $\Serial{\AgentSeq}$ as
the cardinality-$s$ matching of $\Oaf$ in which agent $\Agent_1$ (the
first dictator) is matched to its best match $\Object_1$ in
$\ReachAgain{\Config_0}{\Partial_0}$ where $\Partial_0=\emptyset$,
agent $\Agent_2$ (the second dictator) is matched to its best match
$\Object_2$ in $\ReachAgain{\Config_0}{\Partial_1}$ where
$\Partial_1=\Partial_0+(\Agent_1,\Object_1)$, \ldots, and agent
$\Agent_s$ (the $s$th dictator) is matched to its best match
$\Object_s$ in $\ReachAgain{\Config_0}{\Partial_{s-1}}$ where
$\Partial_{s-1}=\Partial_{s-2}+(\Agent_{s-1},\Object_{s-1})$. Observe
that for any permutation $\AgentSeq$ of the entire set of agents in
$\Oaf$, $\Serial{\AgentSeq}$ is a Pareto-efficient matching of
$\Config_0$.

We iteratively grow a dictator sequence $\AgentSeq$. We find it
convenient to partition the iterations into two phases. The first
phase ends when the current dictator is matched to the center
object. The second phase reduces to solving a collection of disjoint
path problems, one for each branch.

We begin by discussing the design and analysis of the first phase. We
find it useful to introduce the concept of a ``nice pair'' for OAF
$\Oaf$. For any configuration $\Config$ of the form
$(\Oaf,\Matching)$ and any matching $\Partial$ of $\Oaf$, we say that
the pair $(\Config,\Partial)$ is nice for $\Oaf$ if the following
conditions hold:
\begin{itemize}
\item The set $\ReachAgain{\Config}{\Partial}$ is nonempty.
 
\item For any branch $i$ in $[m]$, there is a (possibly empty)
 sequence of integers $1\leq j_1<\cdots<j_s\leq \ell_i$ such that
 $\Partial(\Config(\BranchObject{i}{t}))=\BranchObject{i}{j_t}$ for
 $1\leq t\leq s$ and $\Config(\BranchObject{i}{t})$ is unmatched in
 $\Partial$ for $s<t\leq\ell_i$. We refer to this (unique) sequence
 as $\Dests{\Config}{\Partial}{i}$.
\end{itemize}

The first phase iteratively updates a dictator sequence $\AgentSeq$, a
configuration $\Config$, and a matching $\Partial$. We initialize the
configuration $\Config$ to $\Config_0$, the initial configuration of
$\Oaf$. We initialize $\AgentSeq$ to the singleton sequence containing
agent $\Config(\Center)$, the first dictator. We use Subroutine~1 of
Section~\ref{sec:gsSubroutines} to determine the best match of agent
$\Config(\Center)$ in $\Reach{\Config}$, call it $\Object$, and we
initialize $\Partial$ to $\{(\Config(\Center),\Object)\}$.

We then execute the while loop described below, which we claim
satisfies the following loop invariant $I$: $(\Config,\Partial)$ is a
nice pair for $\Oaf$, agent $\Config(\Center)$ is matched in
$\Partial$,
$\ReachAgain{\Config}{\Partial}=\ReachAgain{\Config_0}{\Partial}$, and
$\Partial=\Serial{\AgentSeq}$. It is straightforward to verify that
invariant $I$ holds after initialization of $\AgentSeq$, $\Config$,
and $\Partial$. In Section
~\ref{sec:gsInvariant}, we prove that if
$I$ holds at the start of an iteration of the while loop, then $I$ holds
at the end of the iteration.

While $\Partial(\Config(\Center))\not=\Center$, we use the following
steps to update $\AgentSeq$, $\Config$, and $\Partial$.
\begin{enumerate}
\item
Since $\Partial(\Config(\Center))\not=\Center$, object
$\Partial(\Config(\Center))$ is of the form $\BranchObject{i}{j_0}$
for some $i$ in $[m]$ and $j_0$ in $[\ell_i]$. Let $j_1<\cdots<j_s$
denote $\Dests{\Config}{\Partial}{i}$, and let $\Agent$ denote the
agent $\Config(\BranchObject{i}{s+1})$.
 
\item
Append agent $\Agent$ to the dictator sequence $\AgentSeq$.
 
\item
\label{step:max}
Use Subroutine~2 of Section~\ref{sec:gsSubroutines} to set $k$ to the
maximum $j$ such that object $\BranchObject{i}{j}$ is a possible match
of agent $\Agent$ in $\ReachAgain{\Config}{\Partial}$, or to $0$ if no
such $j$ exists.
 
\item
\label{step:center}
If
$\BranchObject{i}{s+1}<_\Agent\cdots<_\Agent\BranchObject{i}{1}<_\Agent\Center$,
then perform the following steps.
\begin{itemize}
\item[(a)]
Let $\Matching$ denote the matching of $\Oaf$ such that
$\Config=(\Oaf,\Matching)$, let $\Matching^*$ denote the matching
obtained from $\Matching$ by applying $s+1$ swaps to move agent
$\Agent$ from object $\BranchObject{i}{s+1}$ to the center object
$\Center$, and let $\Config^*$ denote the configuration
$(\Oaf,\Matching^*)$.
 
\item[(b)]
Use Subroutine~1 of Section~\ref{sec:gsSubroutines} to set $\Object$
to the best match of agent $\Agent$ in
$\ReachAgain{\Config^*}{\Partial}$.
 
\item[(c)] If $k=0$ or $\Object>_\Agent\BranchObject{i}{k}$, then set
 $\Config$ to $\Config^*$, $\Partial$ to $\Partial+(\Agent,\Object)$,
 and $k$ to $-1$.
\end{itemize}
 
\item
\label{step:outward}
If $k>0$, then set $\Partial$ to
$\Partial+(\Agent,\BranchObject{i}{k})$.
\end{enumerate}

Upon termination of the first phase, invariant $I$ holds and
$\Partial(\Config(\Center))=\Center$. Thus, letting $\Config_1$
denote the value of program variable $\Config$ at the end of the first
phase, we know that $(\Config_1,\Partial)$ is a nice pair for $\Oaf$,
$\Partial(\Config_1(\Center))=\Center$,
$\ReachAgain{\Config_1}{\Partial}=\ReachAgain{\Config_0}{\Partial}$,
and $\Partial=\Serial{\AgentSeq}$.

In the second phase, we perform the following computation for each
branch $i$ (in arbitrary order). First, we let $j_1<\cdots<j_s$ denote
$\Dests{\Config_1}{\Partial}{i}$. Second, we perform the following
steps for $j$ ranging from $s+1$ to $\ell_i$.
\begin{enumerate}
\item Let $\Agent$ denote agent $\Config_1(\BranchObject{i}{j})$, and
 append $\Agent$ to $\AgentSeq$.

\item Use Subroutine~3 of Section~\ref{sec:gsSubroutines} to set
 $\Object$ to the best match of agent $\Agent$ in
 $\ReachAgain{\Config_1}{\Partial}$.

\item
\label{step:update}
Set $\Partial$ to $\Partial+(\Agent,\Object)$.
\end{enumerate}
At the end of the second phase, we output the matching $\Partial$.

Let $I'$ denote the invariant
``$\ReachAgain{\Config_1}{\Partial}=\ReachAgain{\Config_0}{\Partial}$
and $\Partial=\Serial{\AgentSeq}$''. Thus invariant $I'$ holds at the
end of the first phase. Moreover, it is easy to see that invariant
$I'$ continues to hold immediately after each execution of
step~\ref{step:update} in the second phase.

Since invariant $I$ holds in the first phase and invariant $I'$ holds
in the second phase, the overall algorithm faithfully implements the
serial dictatorship framework discussed at the beginning of this
section. Thus the algorithm correctly computes a Pareto-efficient
matching for configuration $\Config_0$. In
Section~\ref{sec:gsSubroutines}, we explain how to implement
Subroutines~1, 2, and~3 so that the overall running time of the
algorithm is $O(n^2\log n)$.

\subsection{Polynomial-Time Implementation}
\label{sec:gsSubroutines}

In this section we describe an efficient implementation of the
two-phase algorithm presented in Section~\ref{sec:gsPeAlg}.  Our
description of the first phase make use of Subroutines~1 and~2, while
our description of the second phase makes use of Subroutine~3. Below
we discuss how to implement Subroutines~1, 2, and~3 efficiently.  Our
analysis of the time complexity of these subroutines assumes that a
certain preprocessing phase has been performed. Specifically, for each
agent $\Agent$ in $\Oaf$, we precompute the set of all objects
$\Object$ such that the sequence of objects
$\Config_0(\Agent)=b_1,\ldots,b_k=b$ on the unique simple path from
$\Config_0(\Agent)$ (the initial object of agent $\Agent$) to
$\Object$ in $G$ satisfies $b_1<_\Agent\cdots<_\Agent b_k$.  It is
straightforward to compute each such set in $O(n)$ time, and hence the
overall time complexity of the preprocessing phase is $O(n^2)$.

We now describe Subroutine~1. The input to Subroutine~1 is a nice pair
$(\Config,\Partial)$ for $\Oaf$ such that agent
$\Agent=\Config(\Center)$ is unmatched in $\Partial$. The output of
Subroutine~1 is the best match of $\Agent$ in
$\ReachAgain{\Config}{\Partial}$. Subroutine~1 works by considering
each branch $i$ in turn to compute the best branch-$i$ match of
$\Agent$ in $\ReachAgain{\Config}{\Partial}$.  For a fixed $i$ in
$[m]$, the latter problem can be solved as follows. Consider the path
$P$ of objects consisting of the center object $\Center$ plus branch
$i$. By restricting the generalized star configuration $\Config$ to
path $P$, we obtain a path configuration $\Config_P$.  Similarly, by
restricting the matching $\Partial$ to the agents associated with path
$P$, we obtain a matching $\Partial_P$ defined on path $P$. For any
given $j$ in $[\ell_i]$, it is easy to argue that object
$\BranchObject{i}{j}$ is a possible match of $\Agent$ in
$\ReachAgain{\Config}{\Partial}$ if and only if object
$\BranchObject{i}{j}$ is a possible match of $\Agent$ in
$\ReachAgain{\Config_P}{\Partial_P}$. Moreover, we can determine
whether $\BranchObject{i}{j}$ is a possible match of $\Agent$ in
$\ReachAgain{\Config_P}{\Partial_P}$ by performing at most one call to
Algorithm~\ref{alg:greedyPath} on path $P$. Given the results of the
preprocessing phase, the additional time complexity required to
determine whether $\BranchObject{i}{j}$ is a possible match of
$\Agent$ in $\ReachAgain{\Config_P}{\Partial_P}$ is $O(\ell_i)$. Using
binary search, we can determine the maximum $j$ (if any) such that
$\BranchObject{i}{j}$ is a possible match of $\Agent$ in
$\ReachAgain{\Config_P}{\Partial_P}$ in $O(\ell_i\log\ell_i)$
time. (Remark: Letting $j_1<\cdots<j_s$ denote
$\Dests{\Config}{\Partial}{i}$, we can restrict the binary search to
the interval $\{1,\ldots,j_1-1\}$ if $s>0$.)  Thus we can determine
the best branch-$i$ match of $\Agent$ (if any) in
$\ReachAgain{\Config}{\Partial}$ in $O(\ell_i\log\ell_i)$ time, and
hence we can determine the best match of $\Agent$ in
$\ReachAgain{\Config}{\Partial}$ in $O(n\log n)$ time.

We now describe Subroutine~2. The input to Subroutine~2 is a nice pair
$(\Config,\Partial)$ for $\Oaf$ and an integer $i$ in $[m]$ such that
$s<\ell_i$ where $j_1<\cdots<j_s$ denotes
$\Dests{\Config}{\Partial}{i}$. The output $k$ of Subroutine~2 is the
maximum $j$ such that $\BranchObject{i}{j}$ is a possible match of
agent $\Agent=\BranchObject{i}{s+1}$ in
$\ReachAgain{\Config}{\Partial}$, or $0$ if no such $j$ exists.  As in
Subroutine~1, we can reduce this task to a path problem. In the
present case, we can restrict $\Config$ and $\Partial$ to branch $i$,
that is, we do not need to consider the extended path that includes
the center object. Moreover, if $s>0$ we can restrict the binary
search for $j$ to the interval $\{j_s+1,\ldots,\ell_i\}$. Each
iteration of the binary search involves a single call to
Algorithm~\ref{alg:greedyPath}. Given the results of the preprocessing
phase, the additional time complexity required for each such call is
$O(\ell_i)$.  Taking into account the binary search, this approach
yields a time complexity of $O(\ell_i\log\ell_i)$.

Having discussed Subroutines~1 and~2, we can now establish an upper
bound on the time complexity of the first phase.  The first phase
performs at most $n$ iterations.  The worst-case running time of each
iteration is dominated by the cost of a possible call to Subroutine~1,
and hence is $O(n\log n)$. Thus the overall running time of the first
phase is $O(n^2\log n)$.

We now discuss Subroutine~3 and the time complexity of the second
phase.  Since $\Partial(\Config(\Center))=\Center$ throughout the
second phase, the sequence of calls to Subroutine~3 that we make for a
given value of $i$ can be resolved by restricting attention to the
branch-$i$ objects and their matched agents under configuration
$\Config_1$. Because $(\Config_1,\Partial)$ is a nice pair for $\Oaf$
at the end of the first phase, and because we iterate over increasing
values of $j$ from $s+1$ to $\ell_i$, the approach of
Section~\ref{sec:pathPareto} can be used to implement each successive
call to Subroutine~3 in $O(\ell_i)$ time. Thus the time required to
process branch $i$ is $O(\ell_i^2)$, and the overall time complexity
of the second phase is $\sum_{i\in[m]}O(\ell_i^2)=O(n^2)$.

Since the time complexity of the preprocessing phase is $O(n^2)$, the
time complexity of the first phase is $O(n^2\log n)$, and the time
complexity of the second phase is $O(n^2)$, we conclude that the
overall time complexity of the algorithm is $O(n^2\log n)$.

\subsection{The First Phase Invariant}
\label{sec:gsInvariant}

The following sequence of lemmas pertain to an arbitrary iteration in
the first phase.  We assume that invariant $I$ holds before the
iteration, and we seek to prove that invariant $I$ holds after the
iteration.  We use the symbols $\AgentSeq$, $\Config$, and $\Partial$
(resp., $\AgentSeq'$, $\Config'$, $\Partial'$) to refer to the values
of the corresponding program variables at the start (resp., end) of
the iteration. We use the symbols $i$, $j_0$, $s$, $j_1,\ldots,j_s$,
$\Agent$, $\Matching$, $\Matching^*$, and $\Config^*$ to refer to the
corresponding program variables; these variables do not change in
value during the iteration.  We use the symbol $k$ to refer to the
initial value of the corresponding program variable. The value of the
program variable $k$ can change at most once (see
step~\ref{step:center}(c)). We use the symbol $k'$ to refer to the
value of program variable $k$ at the end of the iteration.

\begin{lemma}\label{lem:PE_GS_1}
We have $j_t>t$ for all $t$ in $[s]$.
\end{lemma}

\begin{proof}
Since $(\Config,\Partial)$ is nice, we have
$\ReachAgain{\Config}{\Partial}\not=\emptyset$. Let $\Matching'$
denote a matching in $\ReachAgain{\Config}{\Partial}$.  Since
$\Partial(\Config(\BranchObject{i}{1}))=\BranchObject{i}{j_1}$ and
agent $\Config(\Center)$ cannot overtake agent
$\Config(\BranchObject{i}{1})$ as we transform the matching associated
with $\Config$ to $\Matching'$, we have $1\leq j_0<j_1$.  Since
$j_1<\cdots<j_s$, we deduce that $j_2>2$, \ldots, $j_s>s$.
\end{proof}

It follows from Lemma~\ref{lem:PE_GS_1} that $s<\ell_i$ and hence that
agent $\Agent$ is well-defined. Agent $\Agent$ serves as the dictator
in this iteration.

Throughout the remainder of this section, let $R$ denote the set of
all matchings in $\ReachAgain{\Config}{\Partial}$ that match $\Agent$
to an object in branch $i$, and let $J$ denote the set of all integers
$j$ in $[\ell_i]$ such that object $\BranchObject{i}{j}$ is a possible
match of agent $\Agent$ in $R$.  Thus $k$ is the maximum integer in
$J$, or $0$ if $J$ is empty.

\begin{lemma}\label{lem:PE_GS_2}
The following statements hold: (1) if
$\BranchObject{i}{s+1}<_\Agent\cdots<_\Agent\BranchObject{i}{1}<_\Agent\Center$
does not hold, then $\ReachAgain{\Config}{\Partial}=R$; (2) if $j$
belongs to $J$ and $s=0$, then $j>1$; (3) if $j$ belongs to $J$ and
$s>0$ then $j>j_s$; (4) if $k>0$ then $\BranchObject{i}{k}$ is the
best match of agent $\Agent$ in $R$; (5) if $k=0$ then
$\BranchObject{i}{s+1}<_\Agent\cdots<_\Agent\BranchObject{i}{1}<_\Agent\Center$
holds.
\end{lemma}

\begin{proof}
We begin by proving part~(1). Assume that
$\BranchObject{i}{s+1}<_\Agent\cdots<_\Agent\BranchObject{i}{1}<_\Agent\Center$
does not hold.  It follows agent $\Agent$ cannot be matched to an
object outside of branch $i$ in $\ReachAgain{\Config}{\Partial}$.
Hence $\ReachAgain{\Config}{\Partial}=R$, as required.

To establish parts~(2) through~(5), we consider two cases.

Case~1: $J=\emptyset$. In this case, $k=0$ and hence parts~(2), (3),
and~(4) hold vacuously. It remains to prove part~(5).  Assume for the
sake of contradiction that
$\BranchObject{i}{s+1}<_\Agent\cdots<_\Agent\BranchObject{i}{1}<_\Agent\Center$
does not hold.  It follows from part~(1) that
$\ReachAgain{\Config}{\Partial}=R=\emptyset$, contradicting invariant
$I$.

Case~2: $J\not=\emptyset$. Thus $k>0$ and hence part~(5) holds
vacuously.  In the proofs of parts~(2), (3), and~(4) below, let $j$
belong to $J$, and let $\Matching'$ be a matching in $R$ such that
$\Matching'(\Agent)=\BranchObject{i}{j}$.

We first argue that part~(2) holds. Assume for the sake of
contradiction that $s=0$ and $j=1$. Thus
$\Matching'(\Agent)=\Config(\Agent)=\BranchObject{i}{1}$, and hence
$j_0>1$. It follows that agent $\Config(\Center)$ overtakes the
stationary agent $\Agent$ as we transform the matching associated with
$\Config$ to $\Matching'$, a contradiction.

Now we argue that part~(3) holds.  Assume that $s>0$.  We begin
by proving that $j\geq s+1$.  Assume for the sake of contradiction
that $j<s+1$. Since $(\Config,\Partial)$ is nice, $\Agent$ is the
closest branch-$i$ agent to the center that is ``inward-moving'' in
the sense that $\Matching'(\Agent)$ is closer to the center than
$\Config(\Agent)$.  Since no inward-moving agent can overtake another
inward-moving agent, and since agent $\Config(\Center)$ moves into
branch $i$, we deduce that agent $\Agent$ moves out of branch $i$, a
contradiction since $j$ belongs to $J$. Thus $j\geq s+1$. Now we argue
that $j>j_s$.  Lemma~\ref{lem:PE_GS_1} implies that agent
$\Config(\BranchObject{i}{s})$ moves outward to object
$\BranchObject{i}{j_s}$.  Since $j\geq s+1$, agent $\Agent$ is either
stationary or outward-moving, and hence cannot be overtaken by the
outward-moving agent $\Config(\BranchObject{i}{s})$. It follows that
$j>j_s$, as required.

Now we argue that part~(4) holds. Since $j>j_s$ and
Lemma~\ref{lem:PE_GS_1} implies $j_s>s$, we have $j>s+1$.  Thus agent
$\Agent$ is outward-moving on branch $i$. It follows that
$\BranchObject{i}{k}$ is the best match of agent $\Agent$ in $R$, as
required.
\end{proof}

\begin{lemma}\label{lem:PE_GS_3}
Assume that
$\BranchObject{i}{s+1}<_\Agent\cdots<_\Agent\BranchObject{i}{1}<_\Agent\Center$
holds. Then $\Swaps{\Oaf}{\Matching}{\Matching^*}$ and
\[
\ReachAgain{\Config}{\Partial}\setminus R
=\ReachAgain{\Config^*}{\Partial}\not=\emptyset.
\]
\end{lemma}

\begin{proof}
Since $(\Config,\Partial)$ is nice and
$\BranchObject{i}{s+1}<_\Agent\cdots<_\Agent\BranchObject{i}{1}<_\Agent\Center$
holds, Lemma~\ref{lem:PE_GS_1} implies that the $s+1$ exchanges used to transform
$\Matching$ to $\Matching^*$ are all Pareto-improving.  Hence
$\Swaps{\Oaf}{\Matching}{\Matching^*}$,

The set $\ReachAgain{\Config^*}{\Partial}$ is nonempty since it
includes $\Matching^*$. It remains to prove that
$\ReachAgain{\Config}{\Partial}\setminus
R=\ReachAgain{\Config^*}{\Partial}$.

We first argue that
$\ReachAgain{\Config^*}{\Partial}\subseteq\ReachAgain{\Config}{\Partial}\setminus
R$. Let $\Matching^{**}$ belong to $\ReachAgain{\Config^*}{\Partial}$.
Thus $\Swaps{\Oaf}{\Matching^*}{\Matching^{**}}$.  Since
$\Swaps{\Oaf}{\Matching}{\Matching^*}$, we conclude that
$\Swaps{\Oaf}{\Matching}{\Matching^{**}$ and hence $\Matching^{**}$
belongs to $\ReachAgain{\Config}{\Partial}$.  Since
$\Config^*(\Center)=\Agent$ and $\BranchObject{i}{1}<_\Agent\Center$,
we deduce that $\Matching^{**}(\Agent)$ does not belong to branch
$i$. Since $\Matching^{**}$ belongs to
$\ReachAgain{\Config}{\Partial}$ and $\Matching^{**}(\Agent)$ does not
belong to branch $i$, we conclude that $\Matching^{**}$ belongs to
$\ReachAgain{\Config}{\Partial}\setminus R$.

Now we argue that $\ReachAgain{\Config}{\Partial}\setminus
R\subseteq\ReachAgain{\Config^*}{\Partial}$. Let $\Matching^{**}$
belong to $\ReachAgain{\Config^*}{\Partial}\setminus R$. From our
discussion of the reachable matching problem on trees in Section~
\ref{sec:RM_tree},} we can obtain a sequence of 
swaps
that transforms $\Matching$ to $\Matching^{**}$ by repeatedly
performing any 
swap that moves the two participating
agents closer to their matched objects under $\Matching^{**}$. Since
$\Matching^{**}(\Agent)$ does not belong to branch $i$, this means
that we can begin by performing the $s+1$ 
swaps that
transform $\Matching$ to $\Matching^*$.  It follows that
$\Matching^{**}$ belongs to $\ReachAgain{\Config^*}{\Partial}$, as
required.
\end{proof}

Lemma~\ref{lem:PE_GS_3} implies that if
line~\ref{step:center}(b) is executed, then
$\ReachAgain{\Config^*}{\Partial}\not=\emptyset$, and hence object
$\Object$ is well-defined.

\begin{lemma}\label{lem:PE_GS_4}
Assume that
$\BranchObject{i}{s+1}<_\Agent\cdots<_\Agent\BranchObject{i}{1}<_\Agent\Center$
holds.  If $k=0$ or $\Object>_\Agent\BranchObject{i}{k}$, then
$\Object$ is the best match of agent $\Agent$ in
$\ReachAgain{\Config}{\Partial}$ and
\[
\ReachAgain{\Config}{\Partial'}=\ReachAgain{\Config'}{\Partial'}.
\]
Otherwise, $\BranchObject{i}{k}$ is the best match of agent $\Agent$
in $\ReachAgain{\Config}{\Partial}$.
\end{lemma}

\begin{proof}
We consider two cases.

Case~1: $k=0$. In this case, $R=\emptyset$ and
$\Config'=\Config^*$. Hence Lemma~\ref{lem:PE_GS_3} implies
$\ReachAgain{\Config^*}{\Partial}=\ReachAgain{\Config}{\Partial}$. Thus
object $\Object$ is the best match of agent $\Agent$ in
$\ReachAgain{\Config}{\Partial}$ and
$\ReachAgain{\Config}{\Partial'}=\ReachAgain{\Config'}{\Partial'}$.

Case~2: $k>0$. Since $k>0$, object $\BranchObject{i}{k}$ is the best
match of agent $\Agent$ in $R$. Lemma~\ref{lem:PE_GS_3} implies that
$\Object$ is the best match of agent $\Agent$ in
$\ReachAgain{\Config^*}{\Partial}=\ReachAgain{\Config}{\Partial}\setminus
R$. We consider two subcases.

Case~2.1: $\Object>_\Agent\BranchObject{i}{k}$.  Thus $\Object$ is the
best match of agent $\Agent$ in $\ReachAgain{\Config}{\Partial}$,
$\Config'=\Config^*$, $\Partial'=\Partial+(\Agent,\Object)$, and
$\ReachAgain{\Config}{\Partial'}=\ReachAgain{\Config'}{\Partial'}$.

Case~2.2: $\Object<_\Agent\BranchObject{i}{k}$. Thus object
$\BranchObject{i}{k}$ is the best match of agent $\Agent$ in
$\ReachAgain{\Config}{\Partial}$.
\end{proof}

%
%

The next lemma establishes that invariant $I$ holds at the end of the
iteration.

\begin{lemma}
\label{lem:PE_GS_5}
At the end of the iteration, $(\Config',\Partial')$ is a nice pair for
$\Oaf$, agent $\Config'(\Center)$ is matched in $\Partial'$,
$\ReachAgain{\Config'}{\Partial'}=\ReachAgain{\Config_0}{\Partial'}$,
and $\Partial'=\Serial{\AgentSeq'}$.
\end{lemma}

\begin{proof}
Observe that either $k'=-1$ or $k'=k$. Below we consider these two
cases separately.

Case~1: $k'=-1$. In this case, step~\ref{step:center}(c)
is executed and the associated if condition evaluates to
true. Furthermore, the if condition associated with
step~\ref{step:outward} evaluates to false. Thus $\Config'=\Config^*$
and $\Partial'=\Partial+(\Agent,\Object)$.  It is straightforward to
verify that the pair $(\Config',\Partial')$ is nice for $\Oaf$ with
$\Dests{\Config'}{\Partial'}{i}$ equal to $j_0<\cdots<j_s$, and that
$\Config'(\Center)=\Agent$ is matched in $\Partial'$.  By
Lemma~\ref{lem:PE_GS_4}, object $\Object$ is the best match of agent
$\Agent$ in $\ReachAgain{\Config}{\Partial}$.  Since invariant $I$
implies $\Partial=\Serial{\AgentSeq}$ and
$\ReachAgain{\Config}{\Partial}=\ReachAgain{\Config_0}{\Partial}$, we
deduce that $\Partial'=\Serial{\AgentSeq'}$.  By
Lemma~\ref{lem:PE_GS_4},
$\ReachAgain{\Config}{\Partial'}=\ReachAgain{\Config'}{\Partial'}$.
Invariant $I$ implies
$\ReachAgain{\Config}{\Partial}=\ReachAgain{\Config_0}{\Partial}$ and
hence
$\ReachAgain{\Config}{\Partial'}=\ReachAgain{\Config_0}{\Partial'}$.
We conclude that
$\ReachAgain{\Config'}{\Partial'}=\ReachAgain{\Config_0}{\Partial'}$,
as required.

Case~2: $k'=k$. We begin by proving that $k>0$. Assume for the sake of
contradiction that $k=0$.  Part~(5) of Lemma~\ref{lem:PE_GS_2} implies
that
$\BranchObject{i}{s+1}<_\Agent\cdots<_\Agent\BranchObject{i}{1}<_\Agent\Center$
holds. Hence step~\ref{step:center}(c) is executed, and since $k=0$,
the associated if condition evaluates to true. Hence $k'=-1$,
contradicting the Case~2 assumption.

Since $k'=k>0$, we deduce that $\Config'=\Config$ and
$\Partial'=\Partial+(\Agent,\BranchObject{i}{k})$. It is
straightforward to verify that $(\Config',\Partial')$ is nice for
$\Oaf$ with $\Dests{\Config'}{\Partial'}{i}$ equal to
$j_1<\cdots<j_s<k$. Since $\Config'=\Config$ and invariant $I$ holds,
we deduce that agent $\Config'(\Center)$ is matched in $\Partial'$.

We claim that $\BranchObject{i}{k}$ is the best match of agent
$\Agent$ in $\ReachAgain{\Config}{\Partial}$. If
$\BranchObject{i}{s+1}<_\Agent\cdots<_\Agent\BranchObject{i}{1}<_\Agent\Center$
does not hold, the claim follows from parts~(1) and~(4) of
Lemma~\ref{lem:PE_GS_2}, so assume that
$\BranchObject{i}{s+1}<_\Agent\cdots<_\Agent\BranchObject{i}{1}<_\Agent\Center$
holds. Thus the if condition of step~\ref{step:center} evaluates to
true, and hence the if condition of
step~\ref{step:center}(c) is evaluated.  Since $k'=k>0$,
the latter if condition evaluates to false, and hence
Lemma~\ref{lem:PE_GS_4} implies that the claim holds.

Invariant $I$ implies
$\ReachAgain{\Config}{\Partial}=\ReachAgain{\Config_0}{\Partial}$ and
$\Partial=\Serial{\AgentSeq}$. Thus the claim of the preceding
paragraph implies $\Partial'=\Serial{\AgentSeq'}$.

Since $\Config'=\Config$, invariant $I$ implies
$\ReachAgain{\Config'}{\Partial}=\ReachAgain{\Config_0}{\Partial}$ and
hence
$\ReachAgain{\Config'}{\Partial'}=\ReachAgain{\Config_0}{\Partial'}$.
\end{proof}

	\makeatletter
\renewcommand{\boxed}[1]{\text{\fboxsep=.2em\fbox{\m@th$\displaystyle#1$}}}
\makeatother







\newcommand{\PropositionalFormula}{f}
\newcommand{\Uagent}{\hat{u}}
\newcommand{\Uagents}{\hat{U}}
\newcommand{\Vagent}{\hat{v}}
\newcommand{\Vagents}{\hat{V}}
\newcommand{\Wagent}{\hat{w}}
\newcommand{\Wagents}{\hat{W}}
\newcommand{\uobject}{u}
\newcommand{\uobjects}{U}
\newcommand{\vobject}{v}
\newcommand{\vobjects}{V}
\newcommand{\wobject}{w}
\newcommand{\wobjects}{W}
\newcommand{\Xagent}{\hat{x}}
\newcommand{\Xagents}{\hat{X}}
\newcommand{\Tagent}{\hat{t}}
\newcommand{\xobject}{x}
\newcommand{\xobjects}{X}
\newcommand{\tobject}{t}
\newcommand{\InitMatching}{\Matching_0}




\section{NP-Completeness of Reachable Object on Cliques}
\label{sec:NPC_RO_clique}

It is easy to see that the reachable object on cliques problem belongs to NP. 
In this section, we prove that the problem is NP-hard by
presenting a polynomial-time reduction from the known NP-complete problem 2P1N-SAT to reachable object on cliques.

An instance of 2P1N-SAT is a propositional formula $\PropositionalFormula$ over $n$
variables $x_1,\ldots,x_n$ with the following properties: $f$ is the
conjunction of a number of clauses, where each clause is the
disjunction of a number of literals, and each literal is either a
variable or the negation of a variable; each variable occurs exactly
three times in $f$, once in each of three distinct clauses; each
variable $x_i$ occurs twice as a positive literal (i.e., $x_i$) and
once as a negative literal (i.e., $\neg x_i$).

Throughout the remainder of Section~\ref{sec:NPC_RO_clique}, let $\PropositionalFormula$ denote a given instance
of 2P1N-SAT with $n$ variables $x_1,\ldots,x_n$ and $m$ clauses
$C_1,\ldots,C_m$.

In Section~\ref{subsec:NPC_RO_clique_reduction}, we describe a polynomial-time procedure for
transforming $\PropositionalFormula$ into an instance $I$ of reachable object on cliques.
In Section~\ref{sec:ROClique_Correctness}, we prove that $\PropositionalFormula$ is a positive instance of 2P1N-SAT
if and only if $I$ is a positive instance of reachable object on cliques.

\subsection{Description of the Reduction}
\label{subsec:NPC_RO_clique_reduction}

We now describe how we transform a 2P1N-SAT instance $\PropositionalFormula$ into a
corresponding instance $I$ of reachable object on cliques.
For each variable $x_i$ in $\PropositionalFormula$, there are two agents $\Xagent_i^1$ and $\Xagent_i^2$ in $I$.
For each clause $C_j$ in $f$, there are three agents $\Uagent_j, \Vagent_j$ and $\Wagent_j$ in $I$.
Note that the name we use to refer to
each agent in $I$ includes a hat symbol. We adopt the convention that
if the hat symbol is removed from the name of such an agent, we obtain
the name of the initial endowment of that agent.
For example, agents $\Uagent_j$ and $\Xagent^1_i$ are initially endowed with objects $\uobject_j$ and $\xobject^1_i$, respectively. For convenience, we define $\Uagents$ as the set of agents $\{\Uagent_j \mid j \in [m]\}$, and we define $\uobjects$ as the set of objects $\{\uobject_j \mid j \in [m]\}$. We define $\Vagents$, $\vobjects$, $\Wagents$, $\wobjects$, $\Xagents$, and $\xobjects$ similarly.

Let $\Agents$ (resp., $\Objects$) denote the set of all agents (resp., objects) in $I$.
Let $N$ denote $|\Objects|$. Thus $N=3m+2n$.  Let $K_N$ denote
the complete graph with vertex set $B$, and let $E$ denote the edge set of $K_N$.
Let $\InitMatching$ denote the initial matching of agents with objects.

Below we describe the preferences $\Prefs$ of the agents
in $\Agents$ over the objects in $\Objects$.
Let $\Config = (\Oaf, \InitMatching)$ denote the initial configuration of $\Instance$, where $\Oaf = (\Agents, \Objects, \Prefs, \Edges)$. Instance $I$ asks the following
reachability question: Is there a matching $\Matching$ in $\Reach{\Config}$ such that $\Matching(\Wagent_m) = \uobject_1$?


Let variable $x_i$ appear in clauses $C_{p^1_i}$ and $C_{p^2_i}$ as a positive literal, where $p_i^1 < p_i^2$, and in clause $C_{n_i}$ as a negative literal. The definition of 2P1N-SAT implies that $p_i^1, p_i^2,$ and $n_i$ are distinct. 

Below we list the preferences of each agent in $\Agents$. In doing so, we specify only the objects that an agent prefers to its initial endowment; the order of the remaining objects is immaterial. The initial endowment is shown in a box.
For any $i$ in $[n]$, the preferences of agent $\Xagent^1_i$ are
\begin{equation*}
    \Xagent^1_i : \xobject^2_i \Prefs \wobject_{p^1_i} \Prefs \vobject_{p^1_i} \Prefs \boxed{x^1_i}
\end{equation*}
and the preferences of agent $\Xagent^2_i$ are 
\begin{equation*}
    \Xagent^2_i: \wobject_{n_i} \Prefs \vobject_{n_i} \Prefs \wobject_{p^2_i} \Prefs \vobject_{p^2_i} \Prefs \xobject^1_i \Prefs \boxed{\xobject^2_i}
\end{equation*}
if $n_i < p^2_i$,
and are
\begin{equation*}
    \Xagent^2_i: \wobject_{p^2_i} \Prefs \vobject_{p^2_i} \Prefs \wobject_{n_i} \Prefs \vobject_{n_i} \Prefs \xobject^1_i \Prefs \boxed{\xobject^2_i}
\end{equation*}
otherwise.
    


For any $j$ in $[m]$, the preferences of agent $\Uagent_j$ are
\begin{equation*}
    \Uagent_j : \vobject_j \Prefs \boxed{\uobject_j}.
\end{equation*}

For any $j$ in $[m-1]$, the preferences of agent $\Wagent_j$ are
\begin{equation*}
    \Wagent_j : \uobject_{j+1} \Prefs \boxed{\wobject_j}.
\end{equation*}

The preferences of agent $\Wagent_m$ are
\begin{equation*}
    \Wagent_m : \uobject_1 \Prefs \vobject_1 \Prefs \wobject_1 \Prefs \uobject_2 \Prefs \vobject_2 \Prefs \wobject_2 \Prefs \dots \Prefs \uobject_m \Prefs \vobject_m \Prefs \boxed{\wobject_m}.
\end{equation*}

For any $j$ in $[m]$, the preferences of agent $\Vagent_j$ are
\begin{equation*}
    \Vagent_j : \{\xobject^1_i \mid j \in \{p^1_i, n_i\}\} \cup \{\xobject^2_i \mid j = p^2_i\} \Prefs \boxed{\vobject_j},
\end{equation*}
where the set of objects preceding $\vobject_j$ may be ordered arbitrarily.

This completes the description of the reachable object on cliques instance $\Instance$.

\subsection{Correctness of the Reduction}
\label{sec:ROClique_Correctness}

In this section, we prove that $\PropositionalFormula$ is a positive instance of 2P1N-SAT if and only if $\Instance$ is a positive instance of reachable object on cliques. Lemma~\ref{lem:ifdir} establishes the only if direction. 
Lemmas~\ref{lem:no_right_jump} through~\ref{lem:vm_all_rank} lay the groundwork for Lemma~\ref{lem:onlyifdir}, which establishes the if direction.

For the purposes of our analysis, it is convenient to assign a nonnegative integer rank to each object in $\Objects$, as follows.
For any $j$ in $[m]$, we define $\rank(\uobject_j)$ as $3j-2$,  $\rank(\vobject_j)$ as $3j-1$, and $\rank(\wobject_j)$ as $3j$. The rank of any object in $\xobjects$ is defined to be $0$.

Observation~\ref{obs:x1agents} below can be justified by enumerating all those agents who like object $\xobject^1_i$ at least as well as their initial endowment. 
Observations~\ref{obs:x2agents} and~\ref{obs:vwagents} can be justified in a similar manner.

\begin{ob}
\label{obs:x1agents}
For any $i$ in $[n]$ and any matching $\Matching$ in $\Reach{\Config}$, agent $\Matching^{-1}(\xobject^1_i)$ belongs to $\{\Xagent^1_i, \Xagent^2_i, \Vagent_{p^1_i}, \Vagent_{n_i}\}$.
\end{ob}

\begin{ob}
\label{obs:x2agents}
For any $i$ in $[n]$ and any matching $\Matching$ in $\Reach{\Config}$, agent $\Matching^{-1}(x^2_i)$ belongs to $\{\Xagent^1_i, \Xagent^2_i, \Vagent_{p^2_i}\}$.
\end{ob}

\begin{ob}
\label{obs:vwagents}
For any $j$ in $[m]$ and any matching $\Matching$ in $\Reach{\Config}$, agents $\Matching^{-1}(\vobject_j)$ and $\Matching^{-1}(\wobject_j)$ belong to $\{\Uagent_j, \Vagent_j, \Wagent_m\} \cup A_j$ and $\{\Wagent_j, \Wagent_m\} \cup A_j$, respectively, where $A_j$ denotes 
\begin{equation*}
    \{\Xagent^1_i \mid j = p^1_i\} \cup \{\Xagent^2_i \mid j \in \{p^2_i, n_i\}\}.
\end{equation*}
\end{ob}

\begin{lemma}
\label{lem:ifdir}
Assume that 2P1N-SAT instance $\PropositionalFormula$ is satisfiable. Then there is a matching $\Matching$ in $\Reach{\Config}$ such that $\Matching(\Wagent_m) = \uobject_1$ in the reachable object on cliques instance $I$.
\end{lemma}
\begin{proof}
Let $\sigma : \{x_1, \dots, x_n\} \to \{0,1\}$ denote a satisfying assignment for $\PropositionalFormula$. We specify a sequence of matchings $\Matching_0, \dots, \Matching_{3n+3m-1}$, which depends on $\sigma$, such that (1) $\Matching_{3n+3m-1}(\Wagent_m) = \uobject_1$ and (2) $\Matching_{k-1} = \Matching_k$ or $\Swap{\Oaf}{\Matching_{k-1}}{\Matching_k}$
for all $k$ in $[3n+3m-1]$. 
We obtain this sequence in two phases.

In the first phase, we perform the following three steps for each $i$ from $1$ to $n$.
\begin{enumerate}
    \item If $\sigma(x_i) = 1$ and $\Matching_{3i-3}(\Vagent_{p^1_i}) = \vobject_{p^1_i}$, we set $\Matching_{3i-2}$ to the matching obtained by swapping $\Vagent_{p^1_i}$ with $\Xagent^1_i$ in $\Matching_{3i-3}$.
    Otherwise, we set $\Matching_{3i-2}$ to $\Matching_{3i-3}$.
    \item If $\sigma(x_i) = 1$ and $\Matching_{3i-2}(\Vagent_{p^2_i}) = \vobject_{p^2_i}$, we set $\Matching_{3i-1}$ to the matching obtained by swapping $\Vagent_{p^2_i}$ with $\Xagent^2_i$ in $\Matching_{3i-2}$.
    Otherwise, we set $\Matching_{3i-1}$ to $\Matching_{3i-2}$.
    \item If $\sigma(x_i) = 0$ and $\Matching_{3i-1}(\Vagent_{n_i}) = \vobject_{n_i}$,
    we set $\Matching_{3i}$ to the matching obtained by first swapping $\Xagent^1_i$ with $\Xagent^2_i$, and then swapping $\Vagent_{n_i}$ with $\Xagent^2_i$, in $\Matching_{3i-1}$.
    Otherwise, we set $\Matching_{3i}$ to $\Matching_{3i-1}$.
\end{enumerate}
It is easy to check that all of the swaps in the first phase are valid.

Let $A_j$ be as defined in the statement of Observation~\ref{obs:vwagents}.
We claim that at the end of the first phase, agent $\Matching_{3n}^{-1}(\vobject_j)$ belongs to $\Agents_j$ for all $j$ in $[m]$.
Assume for the sake of contradiction that for some $j$ in $[m]$, agent $\Matching_{3n}^{-1}(\vobject_j)$ does not belong to $\Agents_j$. Since agents $\Uagent_j$ and $\Wagent_m$ do not participate in any swap in the first phase, $\Matching_{3n}(\Uagent_j) = \uobject_j$ and   $\Matching_{3n}(\Wagent_m) = \wobject_m$. Since $\Matching_{3n}^{-1}(\vobject_j)$ does not belong to $\{\Uagent_j, \Wagent_m\} \cup \Agents_j$, Observation~\ref{obs:vwagents} implies that $\Matching_{3n}(\Vagent_j) = \vobject_j$. Let variable $x_i$ satisfy clause $C_j$. The swaps in the $i$th iteration of first phase imply that $\Matching_{3i}^{-1}(\vobject_j) \neq \Vagent_j$.
Since no agent in $\Vagents$ participates in more than one swap in the first phase, we have $\Matching_{3n}^{-1}(\vobject_j) \neq \Vagent_j$, a contradiction since $\Matching_{3n}(\Vagent_j) = \vobject_j$. 
This completes the proof of the claim.

Note that there is a unique object of rank $k$ for each $k$ in $[3m]$.
Using the claim stated above, together with the fact that no agent in $\Uagents \cup \Wagents$ participates in a swap in the first phase, it is easy to verify that for any rank $k$ in $[3m-1]$, there is an agent $\Agent$ such that $\rank(\Matching_{3n}(\Agent)) = k$ and $\Agent$ prefers the object of rank $k+1$ to the object of rank $k$.
The preferences of agent $\Wagent_m$ imply that for any rank $k$ in $[3m-1]$, agent $\Wagent_m$ prefers the object of rank $k$ to the object of rank $k+1$.
Moreover, $\rank(\Matching_{3n}(\Wagent_m)) = 3m$.
In the second phase we perform $3m-1$ swaps, each involving agent $\Wagent_m$.
For $k$ running from $3m-1$ down to $1$, we set $\Matching_{3n+3m-k}$ to the matching obtained by swapping $\Wagent_m$ with the agent matched to the object of rank $k$ in $\Matching_{3n+3m-k-1}$.
The foregoing arguments show that all of these $3m-1$ swaps are valid and $\rank(\Matching_{3n+3m-1}(\Wagent_m)) = 1$. Thus $\Matching_{3n+3m-1}(\Wagent_m) = \uobject_1$.
\end{proof}

Observation~\ref{obs:wmrank} below can be justified by using the preferences of agent $\Wagent_m$ and the fact that if any agent $a$ swaps from object $\Object$ to $\Object'$ then $\Object' \Prefs_a \Object$.
Observations~\ref{obs:vrank} through~\ref{obs:x2rank} can be justified similarly.

\begin{ob}
\label{obs:wmrank}
For any matchings $\Matching_1$ and $\Matching_2$ in $\Reach{\Config}$ such that $\Swap{\Oaf}{\Matching_1}{\Matching_2}$, we have
$\rank(\Matching_2(\Wagent_m)) \le \rank(\Matching_1(\Wagent_m))$.
\end{ob}

\begin{ob}
\label{obs:vrank}
For any $j$ in $[m-1]$ and any matchings $\Matching_1$ and $\Matching_2$ in $\Reach{\Config}$ such that $\Swap{\Oaf}{\Matching_1}{\Matching_2}$, we have 
$\rank(\Matching_2(\Wagent_j)) \le \rank(\Matching_1(\Wagent_j)) + 1$.
\end{ob}

\begin{ob}
\label{obs:v'rank}
For any $j$ in $[m]$ and any matchings $\Matching_1$ and $\Matching_2$ in $\Reach{\Config}$ such that $\Swap{\Oaf}{\Matching_1}{\Matching_2}$, we have 
$\rank(\Matching_2(\Uagent_j)) \le \rank(\Matching_1(\Uagent_j)) + 1$.
\end{ob}

\begin{ob}
\label{obs:urank}
For any $j$ in $[m]$ and any matchings $\Matching_1$ and $\Matching_2$ in $\Reach{\Config}$ such that $\Swap{\Oaf}{\Matching_1}{\Matching_2}$, we have
$\rank(\Matching_2(\Vagent_j)) \le \rank(\Matching_1(\Vagent_j))$.
\end{ob}

\begin{ob}
\label{obs:x1rank}
For any $i$ in $[n]$ and any matchings $\Matching_1$ and $\Matching_2$  in $\Reach{\Config}$ such that $\Swap{\Oaf}{\Matching_1}{\Matching_2}$ and $\Matching_1(\Xagent^1_i)$ belongs to $\Objects \setminus \xobjects$, we have $\rank(\Matching_2(\Xagent^1_i)) \le \rank(\Matching_1(\Xagent^1_i)) + 1$.

\end{ob}

\begin{ob}
\label{obs:x2rank}
For any $i$ in $[n]$ and any matchings $\Matching_1$ and $\Matching_2$ in $\Reach{\Config}$ such that $\Swap{\Oaf}{\Matching_1}{\Matching_2}$ and $\Matching_1(\Xagent^2_i)$ belongs to $\Objects \setminus \xobjects$, we have
$\rank(\Matching_2(\Xagent^2_i)) \le \rank(\Matching_1(\Xagent^2_i)) + 1$.
\end{ob}

Lemma~\ref{lem:no_right_jump} below follows immediately from Observations~\ref{obs:wmrank} through~\ref{obs:x2rank}.
\begin{lemma}
\label{lem:no_right_jump}
For any agent $\Agent$ in $A$ and any matchings $\Matching_1$ and $\Matching_2$ in $\Reach{\Config}$ such that $\Swap{\Oaf}{\Matching_1}{\Matching_2}$ and $\Matching_1(\Agent)$ belongs to $\Objects \setminus \xobjects$, we have
$\rank(\Matching_2(\Agent)) \le \rank(\Matching_1(\Agent)) + 1$.
\end{lemma}


\begin{lemma}
\label{lem:wmrank}
Let $\Matching_1$ and $\Matching_2$ be matchings in $\Reach{\Config}$ such that $\Swap{\Oaf}{\Matching_1}{\Matching_2}$. Then 
\begin{equation*}
    \rank(\Matching_2(\Wagent_m)) \ge \rank(\Matching_1(\Wagent_m)) - 1.
\end{equation*}
\end{lemma}
\begin{proof}
Assume for the sake of contradiction that $\rank(\Matching_2(\Wagent_m)) < \rank(\Matching_1(\Wagent_m)) - 1$.
The preferences of $\Wagent_m$ imply that $\Matching_1(\Wagent_m)$ and $\Matching_2(\Wagent_m)$ belong to $\Objects \setminus \xobjects$.
Thus there is an agent $\Agent$ such that $\Matching_1(\Agent)$ belongs to $\Objects \setminus \xobjects$ and $\rank(\Matching_2(\Agent)) > \rank(\Matching_1(\Agent)) + 1$, contradicting Lemma~\ref{lem:no_right_jump}.
\end{proof}

\begin{lemma}
\label{lem:vprefs}
Let $i$ belong to $[n]$ and $j$ belong to $[m]$. If $\Vagent_j$ prefers $\xobject^1_i$ to $\vobject_j$, then $\Vagent_j$ prefers $\vobject_j$ to $\xobject^2_i$. Similarly, if $\Vagent_j$ prefers $\xobject^2_i$ to $\vobject_j$, then $\Vagent_j$ prefers $\vobject_j$ to $\xobject^1_i$.
\end{lemma}
\begin{proof}
Assume that $\Vagent_j$ prefers $\xobject^1_i$ to $\vobject_j$. The preferences of $\Vagent_j$ imply that $j \in \{p^1_i, n_i\}$. Since $p^1_i, p^2_i,$ and $n_i$ are distinct, $j \neq p^2_i$. 
Hence the preferences of $\Vagent_j$ imply that $\Vagent_j$ prefers $\vobject_j$ to $\xobject^2_i$. We can use a similar argument to prove that if $\Vagent_j$ prefers $\xobject^2_i$ to $\vobject_j$, then $\Vagent_j$ prefers $\vobject_j$ to $\xobject^1_i$.
\end{proof}

\begin{lemma}
\label{lem:x2limit}
Let $\Matching_1$ and $\Matching_2$ be matchings in $\Reach{\Config}$ such that $\Swap{\Oaf}{\Matching_1}{\Matching_2}$ and $\Matching_1 (\Xagent^2_i) = \xobject^2_i$. Then $\Matching_2(\Xagent^2_i)$ belongs to $\{\xobject^1_i, \xobject^2_i, \vobject_{p^2_i}\}$.
\end{lemma}
\begin{proof}
By examining the preferences of $\Xagent^2_i$ we deduce that object $\Matching_2(\Xagent^2_i)$ belongs to 
\begin{equation*}
    \{\wobject_{n_i}, \wobject_{p^2_i}, \vobject_{n_i}, \vobject_{p^2_i}, \xobject^1_i, \xobject^2_i\}.
\end{equation*}
Assume for the sake of contradiction that $\Matching_2(\Xagent^2_i)$ belongs to $\{\wobject_{n_i}, \wobject_{p^2_i}, \vobject_{n_i}\}$. We consider three cases.

Case~1: $\Matching_2(\Xagent^2_i) = \wobject_{n_i}$. Thus $\Matching_1^{-1}(\wobject_{n_i}) = \Matching_2^{-1}(\xobject^2_i)$.
Since $\Matching_2^{-1}(\xobject^2_i) \neq \Xagent^2_i$, Observation~\ref{obs:x2agents} implies that $\Matching_2^{-1}(\xobject^2_i)$ belongs to $\{\Xagent^1_i, \Vagent_{p^2_i}\}$.
We consider two cases.

Case~1.1: $\Matching_2^{-1}(\xobject^2_i) = \Xagent^1_i$. Thus $\Matching_1^{-1}(\wobject_{n_i}) = \Xagent^1_i$.
By examining the preferences of $\Xagent^1_i$, we deduce that $\Matching_1(\Xagent^1_i) \neq \wobject_{n_i}$, a contradiction.

Case~1.2: $\Matching_2^{-1}(\xobject^2_i) = \Vagent_{p^2_i}$. A contradiction follows using a similar argument as in Case~1.1.

Case~2: $\Matching_2(\Xagent^2_i) = \wobject_{p^2_i}$. A contradiction follows using a similar argument as in Case~1.

Case~3: $\Matching_2(\Xagent^2_i) = \vobject_{n_i}$. A contradiction follows using a similar argument as in Case~1.

Thus $\Matching_2(\Xagent^2_i)$ belongs to $\{\xobject^1_i, \xobject^2_i, \vobject_{p^2_i}\}$.
\end{proof}

Throughout the remainder of Section~\ref{sec:ROClique_Correctness}, we say that an agent $a$ is satisfied in a matching $\Matching$ if $\Matching(a)$ is the most preferred object of $a$.
In Lemmas~\ref{lem:x2exclusive} through~\ref{lem:onlyifdir} below, let $\Matching_0, \dots, \Matching_t$ be a sequence of matching such that $\Swap{\Oaf}{\Matching_{s-1}}{\Matching_s}$ for all $s$ in $[t]$, and for each $i$ in $[n]$ let $P(i)$, (resp., $Q(i)$ and $R(i)$) denote the predicate that holds if there is an integer $s$ in $[t]$ such that $\Matching_s(\Xagent^1_i) = \vobject_{p^1_i}$ (resp., $\Matching_s(\Xagent^2_i) = \vobject_{p^2_i}$, $\Matching_s(\Xagent^2_i) = \vobject_{n_i}$). Lemmas~\ref{lem:x2exclusive} and~\ref{lem:x1exclusive} below present useful properties of these predicates. 

\begin{lemma}
\label{lem:x2exclusive}
Let $i$ be an element of $[n]$ such that $R(i)$ holds. Then $Q(i)$ does not hold.
\end{lemma}
\begin{proof}
Let $s$ be an element of $[t]$ such that $\Matching_s(\Xagent^2_i) = \vobject_{n_i}$; such an $s$ exists since $R(i)$ holds.
Assume that $Q(i)$ holds.
Let $s'$ be an element of $[t]$ such that $\Matching_{s'}(\Xagent^2_i) = \vobject_{p^2_i}$; such an $s'$ exists as $Q(i)$ holds. We consider two cases.

Case~1: $s' > s$. Let $s''$ be an integer such that $s \le s'' \le s'$.
The preferences of agent $\Xagent^2_i$ imply that $p^2_i < n_i$ and $\Matching_{s''}(\Xagent^2_i)$ belongs to $\{\wobject_{p^2_i}, \vobject_{p^2_i}, \wobject_{n_i}, \vobject_{n_i}\}$.
Hence $\rank(\Matching_{s''}(\Xagent^2_i))$ belongs to $\{3p^2_i, 3p^2_i-1, 3n_i, 3n_i-1\}$.
Note that $\rank(\Matching_{s}(\Xagent^2_i)) = 3n_i-1$, and $\rank(\Matching_{s'}(\Xagent^2_i)) = 3p^2_i-1$.
Hence there is an $s'''$ such that $s \le s''' < s'$ and $\rank(\Matching_{s'''+1}(\Xagent^2_i)) \le \rank(\Matching_{s'''}(\Xagent^2_i)) - 2$. It follows that there is an agent $\Agent$ such that $\Matching_{s'''}(\Agent)$ belongs to $\Objects \setminus \xobjects$ and $\rank(\Matching_{s'''+1}(\Agent)) \ge \rank(\Matching_{s'''}(\Agent)) + 2$,
contradicting Lemma~\ref{lem:no_right_jump}.

Case~2: $s' < s$. We can derive a contradiction using a similar argument as in Case~1.
\end{proof}

\begin{lemma}
\label{lem:x1exclusive}
Let $i$ be an element of $[n]$ such that $R(i)$ holds. Then $P(i)$ does not hold.
\end{lemma}
\begin{proof}
Let $s$ be an element of $[t]$ such that $\Matching_s(\Xagent^2_i) = \vobject_{n_i}$; such an $s$ exists since $R(i)$ holds.
We begin by proving the following claim: There is an integer $s''$ in $[s-1]$ such that $\Matching_{s''}(\Xagent^2_i) = \xobject^1_i$. 
Assume for the sake of contradiction that there is no $s''$ in $[s-1]$ such that $\Matching_{s''}(\Xagent^2_i) = \xobject^1_i$. 
Let $s''$ be the least index in $[s]$ such that $\Matching_{s''}(\Xagent^2_i) \neq \xobject^2_i$.
Since $\Matching_{s''}(\Xagent^2_i)$ does not belong to $\{\xobject^1_i, \xobject^2_i\}$, Lemma~\ref{lem:x2limit} implies that $\Matching_{s''}(\Xagent^2_i) = \vobject_{p^2_i}$. Thus $Q(i)$ holds, contradicting Lemma~\ref{lem:x2exclusive}.
This completes the proof of the claim.

Having established the claim, we let $s''$ denote the least integer in $[s-1]$ such that $\Matching_{s''}(\Xagent^2_i) = \xobject^1_i$.
The preferences of agent $\Xagent^2_i$ imply that $\Matching_{s''-1}(\Xagent^2_i) = \xobject^2_i$. Let $a$ be the agent $\Matching_{s''-1}^{-1}(\xobject^1_i)$. Since $a \neq \Xagent^2_i$, Observation \ref{obs:x1agents} implies that $a$ belongs to $\{\Xagent^1_i, \Vagent_{p^1_i}$, $\Vagent_{n_i}\}$. We consider two cases.

Case 1: $a \in \{\Vagent_{p^1_i}, \Vagent_{n_i}\}$. 
Lemma~\ref{lem:vprefs} implies that $a$ does not prefer $\xobject^2_i$ to their initially endowment. Hence  $\Matching_{s''}^{-1}(\xobject^2_i) \neq a$, a contradiction.

Case 2: $a = \Xagent^1_i$. Since $\Matching_0(\Xagent^1_i) = \Matching_{s''-1}(\Xagent^1_i) = x^1_i$, we deduce that $\Matching_{s}(\Xagent^1_i) = x^1_i \neq \vobject_{p^1_i}$ for all $s'$ such that $0 \le s' < s''$. Moreover, $\Xagent^1_i$ is satisfied in $\Matching_{s''}$ and hence $\Matching_{s'}(\Xagent^1_i) = x^2_i \neq \vobject_{p^1_i}$ for all $s'$ such that $s'' \le s' \le t$. Hence $\Matching_{s'}(\Xagent^1_i) \neq \vobject_{p^1_i}$ for all $s'$ such that $0 \le s' \le t$. Thus $P(i)$ does not hold.
\end{proof}

\begin{lemma}
\label{lem:vm_all_rank}
Let $j$ belong to $[m]$ and assume that $\Matching_t(\Wagent_m) = u_1$. Then there is an $s$ in $[t-1]$ such that $\Matching_s(\Wagent_m) = \wobject_j$ and $\Matching_{s+1}(\Wagent_m) = \vobject_j$.
\end{lemma}
\begin{proof}
The only object with rank $3j$ (resp., $3j-1$) is $\wobject_j$ (resp., $\vobject_j$).
Since $\rank(\InitMatching(\Wagent_m)) = 3m$ and $\rank(\Matching_t(\Wagent_m)) = 1$, Lemma~\ref{lem:wmrank} implies that for every rank $k$ in $[3m-1]$, there is an integer $s$ in $[t-1]$ such that $\rank(\Matching_s(\Wagent_m)) = k+1$ and $\rank(\Matching_{s+1}(\Wagent_m)) = k$. 
The lemma follows by choosing $k$ to be $3j-1$.
\end{proof}

\begin{lemma}
\label{lem:onlyifdir}
Assume that $\Matching_t(\Wagent_m)=\uobject_1$. Then the 2P1N-SAT
instance $\PropositionalFormula$ is satisfiable.
\end{lemma}
\begin{proof}
We construct an assignment $\sigma : \{x_1, \dots, x_n\} \to
\{0,1\}$ for $\PropositionalFormula$ as follows:
for any $i$ in $[n]$, we set $\sigma(x_i)$ to $1$ if $P(i) \vee Q(i)$ holds, and to $0$ otherwise.

We now show that $\sigma$ satisfies $\PropositionalFormula$. 
Let $j$ belong to $[m]$.
Let $s$ be an element of $[t-1]$ such that $\Matching_s(\Wagent_m) = \wobject_j$ and $\Matching_{s+1}(\Wagent_m) = \vobject_j$; such an $s$ exists by Lemma~\ref{lem:vm_all_rank}.
Thus there is an agent $a$ such that $\Matching_s(a) = \vobject_j$ and $\Matching_{s+1}(a) = \wobject_j$.
Since $a \neq \Wagent_m$, Observation~\ref{obs:vwagents} implies that exactly one of the following three statements holds: (1) $j = p^1_i$ and $\Matching_s(\Xagent^1_i) = \vobject_j$; (2) $j = p^2_i$ and $\Matching_s(\Xagent^2_i) = \vobject_j$; (3) $j = n_i$ and $\Matching_s(\Xagent^2_i) = \vobject_j$. We consider two cases.

Case~1: (1) or (2) holds. Then $P(i)$ or $Q(i)$ holds, respectively. Hence $\sigma(x_i) = 1$. By construction, the variable $x_i$ appears as the positive literal $x_i$ in clause $C_j$. Thus, $C_j$ is satisfied.

Case~2: (3) holds. Then $R(i)$ holds. Lemmas~\ref{lem:x2exclusive} and~\ref{lem:x1exclusive} imply that $P(i)$ and $Q(i)$ do not hold. Hence $\sigma(x_i) = 0$. By construction, the variable $x_i$ appears as the negative literal $\neg x_i$ in clause $C_j$. Thus, $C_j$ is satisfied.

Since $\sigma$ satisfies each clause $C_j$ in $\PropositionalFormula$, $\sigma$ satisfies $\PropositionalFormula$.
\end{proof}

\begin{theorem}
\label{thm:obj-reach-clique}
Reachable object on cliques is NP-complete.
\end{theorem}
\begin{proof}
In Section \ref{subsec:NPC_RO_clique_reduction}, we described a polynomial-time reduction from 2P1N-SAT instance $\PropositionalFormula$ to reachable object on cliques instance $I$. Thus the theorem follows from 
Lemmas~\ref{lem:ifdir} and~\ref{lem:onlyifdir}.
\end{proof}

	\section{Other NP-Completeness and NP-Hardness Results}
\label{sec:hardness}

We prove the remaining NP-completeness and NP-hardness results stated in
Table~\ref{tbl:dozen}. 
Specifically, we show that the reachable object problem on generalized stars is NP-complete, 
reachable matching problem on cliques is NP-complete,
and Pareto-efficient matching problem on cliques is NP-hard. 
%
%
%
%
%
These results are proved by
adapting the corresponding proofs of Gourv{\`e}s et
al.~\cite{gourves2017object} and M{\"u}ller and Bentert~\cite{muller2020reachable} 
for the object-moving model. The most
significant changes are associated with the proof of
the first result. 

%

\subsection{NP-Completeness for Reachable Object on Generalized Stars}
\label{sec:NPC_RO_generalized_star}
\begin{restatable}{theorem}{NpcRO}\label{thm:NPC_RO_generalized_star}
	Reachable object on generalized stars is NP-complete. 
\end{restatable}
\begin{proof}
It is easy to see that reachable object on generalized stars belongs to NP.
We use a reduction from the problem 2P1N-SAT to establish that the
reachable object on generalized stars is NP-complete. In an instance of 2P1N-SAT, we are
given a propositional formula $f$ that is the conjunction of $m$
clauses $C_1,\ldots,C_m$. Each clause $C_i$ is the disjunction of a
number of literals, where each literal is either a variable or the
negation of a variable. The set of variables is $x_1,\ldots,x_n$.  For
each variable $x_j$, the positive literal $x_j$ appears in exactly
two clauses, and the negative literal $\neg x_j$ appears in
exactly one clause. We are asked to determine whether the formula $f$
is satisfiable.

Given such a formula $f$, we construct a corresponding reachable
object on generalized stars instance $I=(\Oaf,\Matching)$ where
$\Oaf=(\Agents,\Objects,\Prefs,\Edges)$ as follows.  We begin by
describing the set of objects $\Objects$.   For each clause index $i$,
there are two objects $c_i'$ and $c_i''$ in $\Objects$.  There is also
a special object $c_0''$ in $\Objects$.  For each variable index $j$,
there are five objects in $\Objects$: ``dummy'' objects $d_j$ and
$d'_j$; an object $n_j$ corresponding to the lone occurrence of the
negative literal $\neg x_j$ in $f$; an object $p_j$ corresponding to
the first occurrence of the positive literal $x_j$ in $f$ (i.e., the
occurrence associated with the lower-indexed clause); an object $p_j'$
corresponding to the other occurrence of the positive literal
$x_j$. Thus there are a total of $2m+5n+1$ objects in $\Objects$.

Observe that our identifier for any given object in $\Objects$
includes a single lowercase letter. By changing this letter to upper
case, we obtain our identifier for the agent initially matched to that
object.  So, for example, agent $C''_0$ is initially matched to object
$c''_0$.

%

We now describe the edge set $\Edges$. Object $c_m''$ is the center object.
We begin by describing $m+n+1$
vertex-disjoint subgraphs of $(\Objects,\Edges)$: a ``clause gadget''
for each clause index $i$, a ``variable gadget'' for each variable
index $j$, and an additional gadget that we call the
``staircase''. Clause gadget $i$ consists of the lone object
$c_i'$.  Variable gadget $j$ is a path of length $3$ containing the
	objects $n_j$, $d_j$, $p_j'$ and $p_j$,  in that order, and the lone object $d_j'$. The
staircase is a path of length $m$ containing the sequence of objects
$c_0'',c_1'',\ldots,c_m''$. We say that object $c_0''$ is at the
bottom of the staircase, and that object $c_m''$ is at the top of the
staircase. The following additional $m+2n$ edges are used to connect
these $m+n+1$ subgraphs into a generalized star: there is an edge from object
$c_i'$ to object $c_m''$ for each clause index $i$; there is two edges
	from object $p_j,d_j'$ to object $c_m''$, respectively,  for each variable index $j$.

With regard to the agent preferences, all that matters is the set
$\Improving{\Instance}{\Agent}$ associated with each agent $\Agent$ in
$\Agents$. It is straightforward to verify that we can choose agent
preferences so that the following conditions are satisfied. First, for
each clause index $i$, we have
$\Improving{\Instance}{C_i'}=\{c'_i\}\cup\{c''_k\mid i-1\leq k\leq
m\}$. Second, for each integer $i$ such that $0\leq i\leq m$, we have
\[
\Improving{\Instance}{C_i''}=\{c''_k\mid i\leq k\leq m\}
\cup\{p_j\mid  j\in[n] \} \cup\{p'_j\mid j\in[n] \}.
\]
Finally, for each variable index $j$, the following properties hold:
$\Improving{\Instance}{D_j}$ is equal to $\{d_j,p'_j,p_j,d'_j, c_m''\}$;
$\Improving{\Instance}{D'_j}$ is equal to $\{d'_j,p_j,p'_j, c_m''\}$;
$\Improving{\Instance}{N_j}$ is equal to $\{n_j,d_j,p'_j,p_j,c_m'',c_i'\}$ where $i$
is the index of the clause that contains the negative literal $x_j$;
 $\Improving{\Instance}{P_j}$ is equal to $\{p_j,p'_j,d_j,c''_m,c'_i\}$ where $i$ is
the index of the clause that contains the first occurrence of the
positive literal $x_j$;
 $\Improving{\Instance}{P'_j}$ is equal to $\{p'_j,d_j,n_j,p_j,c''_m,c'_i\}$ where
$i$ is the index of the clause that contains the second occurrence of
the positive literal $x_j$.

We claim that agent $C_0''$ (which starts out at the bottom of the
staircase) can reach object $c_m''$ (at the top of the staircase) if
and only if $f$ is satisfiable.

We begin by addressing the ``if'' direction of the claim. Assume that
$f$ is satisfiable. Fix a satisfying assignment for $f$, and for each
clause index $i$, let $\ell_i$ denote a literal in $C_i$ that is set
to true by this satisfying assignment. Let $L_i$ denote the agent that
corresponds to literal $\ell_i$, as follows: if $\ell_i$ is the
negative literal $\neg x_j$, then $L_i$ is equal to $N_j$; if $\ell_i$
is the first occurrence of the positive literal $x_j$, then $\ell_i$
is equal to $P_j$; if $\ell_i$ is the second occurrence of the
positive literal $x_j$, then $\ell_i$ is equal to $P'_j$.  Note that
the $L_i$'s are all distinct.

For any variable $x_j$, let $\LLL_j=\{L_i\mid 1\leq i\leq
m\}\cap\{N_j,P_j,P'_j\}$. Note that if $|\LLL_j|\geq 2$ then
$\LLL_j=\{P_j,P_j'\}$.

For each clause index $i$, $1\leq i\leq m$, let $f(i)$ denote the
unique variable index $j$ such that $L_i\in\{N_j,P_j,P'_j\}$.

We now describe how to perform a sequence of swaps that result in
agent $C_0''$ being matched to object $c_m''$. We perform these swaps
in $m$ phases. We will define each phase so that a certain invariant
holds. Specifically, we will ensure that the following conditions hold
after $k$ phases have been completed, $0\leq k\leq m$: (1) the
sequence of agents associated with the $m+1$ staircase objects,
$c_0'',\ldots,c_m''$ is $C'_1,\ldots,C'_k,C''_0,\ldots,C''_{m-k}$; (2)
for any clause index $i$ such that $k<i\leq m$, the agent matched to
object $c'_i$ is $C'_i$; (3) for any variable index $j$ such that
$\LLL_j\not\subseteq \{L_i\mid 1\leq i\leq k\}$, the agents matched to the
objects $n_j$, $d_j$, $p'_j$, and $d'_j$ are $N_j$, $D_j$, $P'_j$, and
$D'_j$, respectively, and the agent matched to object $p_j$ is $P_j$
if $\{L_i\mid 1\leq i\leq k\}\cap\LLL_j=\emptyset$ and belongs to
$\{C''_{m-i+1}\mid 1\leq i\leq k\}$ otherwise.

Notice that if we can prove the invariant holds after $m$ phases, then
condition~(1) of the invariant implies that agent $C''_0$ is matched
to object $c''_m$, as desired.  It is easy to see that the invariant
holds at the outset (i.e., after $0$ phases). Let $k$ be an integer
such that $1\leq k\leq m$, and assume that the invariant holds after
$k-1$ phases. It remains to describe how to implement phase $k$ so
that the claimed invariant holds after phase $k$.


Let $j$ denote $f(k)$.  We implement phase $k$ in three stages.  In
the first stage, we perform swaps within variable gadget $j$ to move
agent $L_k$ to the center object $c_m''$. To see how to do this, consider the
following cases. Notice that Condition~(2) of the invariant implies the  agent matched with the center object $c_m''$ is $C_{m-k+1}''$.

Case~1: $L_k=N_j$. Thus $\LLL_j=\{N_j\}$.
\fu{Demo movement,  will be commented off in the final version\begin{align*}
	&n_j - d_j &- p'_j - &p_j - c_m'' - d'_j \\
	&N_j - D_j &- P'_j - &P_j - C_{m-k+1}'' - D'_j \\
	&P'_j - P_j &-  C_{m-k+1}''- & D'_j -  N_j - D_j
	\end{align*}}
 Condition~(3) of the
invariant implies that the agents matched to the objects $n_j$, $d_j$,
$p'_j$, $p_j$, and $d'_j$ are $N_j$, $D_j$, $P'_j$, $P_j$, and $D'_j$,
respectively. Using $8$ swaps within agents of variable gadget $j$ and agent $C_{m-k+1}''$, we can
rearrange these six agents so that the agents matched to the objects
$n_j$, $d_j$, $p'_j$, $p_j$, $c_m''$, and $d'_j$ are $P'_j$, $P_j$, $C_{m-k+1}''$, $D'_j$,
$N_j$, and $D_j$, respectively. Thus agent $L_k$ is matched to
the center object $c_m''$, as required.

Case~2: $L_k=P_j$.Thus $\{P_j\}\subseteq\LLL_j\subseteq\{P_j,P'_j\}$ \fu{Demo movement,  will be commented off in the final version\begin{align*}
	&n_j - d_j - p'_j - &p_j - c_m'' - &d'_j \\
	&N_j - D_j - P'_j - &P_j - C_{m-k+1}'' - &D'_j \\
	&N_j - D_j - P'_j - &C_{m-k+1}''- P_j  - &D'_j \\
	\end{align*}}
Condition~(3) of the
invariant implies that $L_k$ is  matched to object $p_j$, so  swap  $L_k$ with $C_{m-k+1}''$ as required.

Case~3: $L_k=P'_j$ and $\LLL_j=\{P'_j\}$.  
\fu{Demo movement, will be commented off in the final version\begin{align*}
	&n_j - d_j &- p'_j - &p_j - c_m'' - &d'_j \\
	&N_j - D_j &- P'_j - &P_j - C_{m-k+1}'' - &D'_j \\
	&N_j - D_j &- P_j - &C_{m-k+1}''- P_j'  - &D'_j 
	\end{align*}}

Condition~(3) of the
invariant implies that the agents matched to the objects $n_j$, $d_j$,
$p'_j$, $p_j$, and $d'_j$ are $N_j$, $D_j$, $P'_j$, $P_j$, and $D'_j$,
respectively. By swapping $P'_j$ with $P_j$ and then with $C_{m-k+1}''$, we can ensure that $L_k$
is matched to object $c_m''$, as required.

Case~4: $L_k=P'_j$ and $\LLL_j=\{P_j,P'_j\}$. Thus $\{L_i\mid 1\leq
i<k\}\cap\LLL_j =\{P_j\}$. 
\fu{Demo movement, will be commented off in the final version\begin{align*}
	&n_j - d_j &- p'_j - &p_j - c_m'' - &d'_j \\
	&N_j - D_j &- P'_j - &C_{j'}'' - C_{m-k+1}'' - &D'_j \\
	&N_j - D_j &- C_{j'}'' - &C_{m-k+1}''- P_j'  - &D'_j 
	\end{align*},
	where  $C_{j'}''$ is an agent in  the set
	$\{C''_{m-i+1}\mid 1\leq i<k\}$}
Condition~(3) of the invariant implies
that the agents matched to the objects $n_j$, $d_j$, $p'_j$, $d'_j$,
and $p_j$ are $N_j$, $D_j$, $P'_j$, $D'_j$, and an agent in the set
$\{C''_{m-i+1}\mid 1\leq i<k\}$, respectively. By swapping agent
$P'_j$ with the agent matched to $p_j$ and  then with $C_{m-k+1}''$, we can ensure that $L_k$ is
matched to object $c_m''$, as required.

At the start of the second stage, the agent matched with the center  object $c_m''$
is $L_k$ (due to the first stage), and
the agent matched with object $c'_k$ is $C'_k$ (due to condition~(2)
of the invariant).  In the second stage, we use one swap to rearrange
these two agents so that the agents matched to the objects 
$c''_m$ and $c'_k$ are $C'_k$ and $L_k$.

At the start of the third stage, the sequence of agents associated
with the $m+1$ staircase objects $c_0'',\ldots,c_m''$ is 
$$C'_1,\ldots,C'_{k-1},C''_0,\ldots,C''_{m-k},C'_k,$$ (due to
condition~(1) of the invariant and the second stage).  In the third
stage, we perform $m-k+1$ swaps to move agent $C'_k$ down from the top
of the staircase to object $c''_{k-1}$.

It remains to verify that the invariant holds after phase $k$. The
third stage ensures that condition~(1) of the invariant holds after
phase $k$. Condition~(2) of the invariant holds after phase $k$
because it held before phase $k$ and the only clause gadget involved
in a swap in phase $k$ is clause gadget $k$.  After phase $k$,
condition~(3) of the invariant only makes a nontrivial claim about the
allocation of a variable gadget $j$ for which $L_k=P_j$ and
$\LLL_j=\{P_j,P'_j\}$. In this case, after phase $k$, the agents
matched to the objects $n_j$, $d_j$, $p'_j$, and $d'_j$ are $N_j$,
$D_j$, $P'_j$, $C''_{m-k+1}$, and $D'_j$, respectively, so the
associated claim is satisfied.  Moreover, the claims made in
condition~(3) after phase $k$ that concern other variable gadgets
follow from condition~(3) before phase $k$ since the only variable
gadget involved in any swaps in phase $k$ is variable gadget $j$.

We now address the ``only if'' direction. Assume that a sequence $S$
of valid swaps results in agent $C_0''$ being matched to object
$c_m''$.  Thus there are integers $0=t_0<t_1<\cdots<t_m$ such that for
$0\leq k\leq m$, agent $C_0''$ first becomes matched to object $c_k''$
at ``time'' $t_k$, i.e., immediately after the first $t_k$ swaps of
$S$ have been performed. For any integer $k$ such that $0\leq k\leq
m$, let $Q(k)$ denote the predicate ``at time $t_k$, agent $C_i'$ is
matched to object $c_{i-1}''$ for all $i$ such that $1\leq i\leq
k$''. We now use induction on $k$ to prove that $Q(k)$ holds for
$0\leq k\leq m$. It is easy to see that $Q(0)$ holds. Let $k$ be an
integer such that $1\leq k\leq m$ and assume that $Q(k-1)$ holds. We
need to prove that $Q(k)$ holds.  Since $Q(k-1)$ holds, each agent in
$\{C_i\mid 1\leq i<k\}$ is matched to its favorite object at time
$t_{k-1}$, and hence does not move thereafter.  Thus, to establish
that $Q(k)$ holds, it is sufficient to prove that the agent, call it
$\Agent$, moving from object $c''_k$ to object $c''_{k-1}$ in swap
$t_k$ of $S$ (which moves agent $C''_0$ from object $c''_{k-1}$ to
object $c''_k$) is $C'_k$.  Since $c''_{k-1}$ belongs to
$\Improving{\Instance}{\Agent}$, we deduce that $\Agent$ belongs to
\[
\{C'_i\mid 1\leq i\leq k\}\cup \{C''_i\mid 0\leq i<k\}.
\]
As discussed above, for $1\leq i<k$, agent $C_i'$ is permanently
matched to object $c''_{i-1}$ as of time $t_{k-1}$. It follows that
$\Agent$ does not belong to $\{C'_i\mid 1\leq i<k\}$. It also follows
that each agent in $\{C''_i\mid 1\leq i<k\}$ moved up the staircase
from object $c''_{k-1}$ to object $c''_k$ prior to time $t_{k-1}$, and
hence can never return to object $c''_{k-1}$; thus agent $\Agent$ does
not belong to $\{C''_i\mid 1\leq i<k\}$. Since agent $\Agent$ is not
equal to $C''_0$, we conclude that agent $\Agent$ is equal to $C'_k$,
as required. This completes our proof by induction that $Q(k)$ holds
for $0\leq k\leq m$.

Since $Q(m)$ holds, we know that for each clause index $i$, the
sequence of swaps $S$ causes agent $C'_i$ to move away from its
initial object $c'_i$. For any clause index $i$, let $\Agents_i$ denote
$\{\Agent\in\Agents\mid
c'_i\in\Improving{\Instance}{\Agent}\}-C'_i$. We can only swap agent
$C'_i$ away from its initial object $c'_i$ in favor of some agent in
$\Agents_i$. Our reduction ensures that $\Agents_i$ is equal to the set of
agents corresponding to literals satisfying clause $i$, in the
following sense: for each negative literal $\neg x_j$ appearing in
$C_i$, the agent $N_j$ belongs to $\Agents_i$; for each positive literal
$x_j$ such that the first (resp., second) occurrence of $x_j$ appears
in $C_i$, the agent $P_j$ (resp., $P'_j$) belongs to $\Agents_i$. For
each clause index $i$, let $L_i$ denote the agent in $\Agents_i$ that
swaps with agent $C'_i$ when $C'_i$ moves away from its initial object
$c'_i$, and let $\ell_i$ denote the literal corresponding to $L_i$. We
claim that it is possible to find a truth assignment for $f$ that
simultaneously sets all of the literals $\ell_i$ to true, and thus
satisfies $f$.  To prove this, it suffices to show that for any
variable index $j$, if $N_j$ belongs to $\{L_i\mid 1\leq i\leq m\}$
then $\{L_i\mid 1\leq i\leq m\}\cap\{P_j,P'_j\}=\emptyset$. Below we
prove that the following stronger claim holds: For any variable index
$j$, if a sequence of swaps causes agent $N_j$ to leave variable
gadget $j$ (i.e., to move from object $p_j$ to object $c''_m$) then
agents $P_j$ and $P'_j$ remain in variable gadget $j$ under this
sequence of swaps.

To prove the latter claim, let us fix a sequence of swaps $S$ that
causes agent $N_j$ to leave variable gadget $j$. The only agent
$\Agent\not= N_j$ such that object $n_j$ belongs to
$\Improving{\Instance}{\Agent}$ is $P'_j$. Since $N_j$ moves away from
its initial object $n_j$ under $S$, we deduce that $S$ includes two
swaps moving agent $P'_j$ first to object $d_j$ and then to object
$n_j$. Since object $n_j$ is a leaf, agent $P'_j$ remains matched to
object $n_j$ thereafter. Since agent $N_j$ does not remain matched to
object $d_j$, it is eventually swapped to object $p'_j$. Since agent
$D_j$ has previously moved from object $d_j$ to object $p'_j$, it
cannot move back to object $d_j$. The only agent
$\Agent\not\in\{N_j,D_j,P'_j\}$ such that $d_j$ belongs to
$\Improving{\Instance}{\Agent}$ is agent $P_j$. Thus, the swap that
moves agent $N_j$ from object to $d_j$ to object $p'_j$ moves agent
$P_j$ from object $p'_j$ to object $d_j$. Since agent $P'_j$ is
permanently matched to object $n_j$, we conclude that agent $P_j$ is
permanently matched to object $d_j$. Thus if agent $N_j$ leaves
variable gadget $j$ (indeed, if $N_j$ merely reaches object $p'_j$),
then neither agent $P_j$ nor agent $P'_j$ leaves variable gadget $j$.
\end{proof}

\subsection{NP-Completeness for Reachable Matching on Cliques}
\label{sec:NPC_RM_clique}

We begin by proving in Lemma~\ref{lem:NPC_RM_general} that reachable matching on general graphs is NP-complete.
We use Lemma~\ref{lem:NPC_RM_general} to establish that reachable matching on cliques is also NP-complete.

\begin{restatable}{lemma}{NpcRM}\label{lem:NPC_RM_general}
	Reachable matching on general graphs is NP-complete. 
\end{restatable}

\begin{proof}
It is easy to see that reachable matching on general graphs belongs to
NP.  We use a reduction from reachable object on generalized stars to reachable
matching on general graphs to establish that reachable matching on
general graphs is NP-complete.

Fix an arbitrary reachable object on generalized stars instance $\Instance$. Without loss of
generality, we can assume that the associated configuration $\Config=
(\Oaf, \Matching)$ is such that $\Oaf =  (\Agents,\Objects,\Prefs,\Edges)$, 
$\Agents=\{\Agent_1,\ldots,\Agent_n\}$,
$\Objects=\{\Object_1,\ldots,\Object_n\}$, and
$\Config(\Agent_i)=\Object_i$ for $1\leq i\leq n$. We can also assume
without loss of generality that our goal is to determine whether there
is a matching $\Matching_1$ in $\Reach{\Config}$ such that
$\Matching_1(\Agent_1)=\Object_n$.

Below we describe how to transform reachable object on generalized stars instance
$\Instance$ into a reachable matching on general graphs instance
$\Instance'$ such that $\Instance$ is a positive instance of reachable
object on generalized stars if and only if $\Instance'$ is a positive instance of
reachable matching on general graphs .  The reachable matching on
general graphs instance $\Instance'$ has two associated configurations
$\Config'=(\Oaf',\Matching')$ and $\Config''=(\Oaf',\Matching'')$,
where $\Oaf'=(\Agents',\Objects',\Prefs',\Edges')$.  The set of agents $\Agents'$ is equal to
$\Agents\cup\Agents^*$ where
$\Agents^*=\{\Agent^*_1,\ldots,\Agent^*_n\}$.  The set of objects
$\Objects'$ is equal to $\Objects\cup\Objects^*$ where
$\Objects^*=\{\Object^*_1,\ldots,\Object^*_n\}$. The perfect matching
$\Matching'$ from $\Agents'$ to $\Objects'$ satisfies
$\Matching'(\Agent_i)=\Object_i$ and
$\Matching'(\Agent^*_i)=\Object^*_i$ for $1\leq i\leq n$. The subgraph
of $(\Objects',\Edges')$ induced by the set of objects $\Objects$ is
equal to $(\Objects,\Edges)$.  The subgraph of $(\Objects',\Edges')$
induced by the set of objects $\Objects^*$ is a clique. There are $n$
edges connecting these two subgraphs: there is an edge from object
$\Object_i$ to object $\Object^*_i$ for $1\leq i\leq n$. The agent
preferences $\Prefs'$ are defined as follows.
\begin{itemize}
	\item For any integer $i$ such that $1\leq i\leq n$, the most
	preferred object of agent $\Agent_i^*$ is $\Object_i$, followed by
	object $\Object^*_i$, followed by the remaining objects in
	$\Objects'$ in arbitrary order.
	
	\item The most preferred object of agent $\Agent_1$ is
	$\Object^*_n$, followed by the objects in $\Objects$ in the order
	specified by the preferences of agent $\Agent_1$ under $\Prefs$,
	followed by the objects in $\Objects^*-\Object^*_n$ in arbitrary
	order.
	
	\item The most preferred objects of agent $\Agent_n$ are
	$\Object^*_1,\ldots,\Object^*_{n-1}$, followed by the objects in
	$\Objects$ in the order specified by the preferences of agent
	$\Agent_n$ under $\Prefs$, followed by object $\Object^*_n$.
	
	\item For any integer $i$ such that $1<i<n$, the most preferred
	objects of $\Agent_i$ are $\Object^*_i,\ldots,\Object^*_{n-1}$,
	followed by $\Object^*_{i-1},\ldots,\Object^*_1$, followed by the
	objects in $\Objects$ in the order specified by the preferences of
	agent $\Agent_i$ under $\Prefs$, followed by object $\Object^*_n$.
\end{itemize}
The perfect matching $\Matching''$ associated with configuration $\Config''$
maps each agent in $\Agents'$ to its most preferred object in
$\Objects'$. (Note that $\Matching''$ is a perfect matching from $\Agents'$ to
$\Objects'$, since no two agents in $\Agents'$ share the same most
preferred object.)

It is easy to see that we can construct instance $\Instance'$ in
polynomial time in the size of instance $\Instance$. It remains to
argue that instance $\Instance$ is a positive instance of reachable
object on generalized stars if and only if $\Instance'$ is a positive instance of
reachable matching on general graphs.

We begin by addressing the ``only if'' direction. Assume that instance
$\Instance$ is a positive instance of reachable object on generalized stars. Thus there is a
configuration $\Config_1$ in $\Reach{\Config}$ such that
$\Config_1(\Agent_1)=\Object_n$. Our construction of the agent
preferences therefore ensures the existence of a configuration
$\Config'_1$ in $\Reach{\Config'}$ such that
$\Config'_1(\Agent_i)=\Config_1(\Agent_i)$ and
$\Config'_1(\Agent^*_i)=\Config(\Agent^*_i)=\Object^*_i$ for $1\leq
i\leq n$.

It is easy to check that a swap across edge $(\Object_i,\Object^*_i)$
can be applied to configuration $\Config'_1$ for $1\leq i\leq n$. Let
$\Config'_2$ denote the configuration obtained by applying these $n$
swaps to $\Config'_1$.  Thus $\Config'_2$ belongs to
$\Reach{\Config'}$. Furthermore, it is easy to check that each
agent in $\Agents^*+\Agent_1$ is matched in $\Config'_2$ to its most
preferred object under $\Prefs'$.

Next, we iteratively construct a sequence of $n-1$ configurations
$\Config'_3,\ldots,\Config'_{n+1}$ such that configuration
$\Config'_k$ satisfies the following properties for $3\leq k\leq n+1$:
$\Config'_k$ belongs to $\Reach{\Config'}$; every agent in
$\Agents\cup\{\Agent_1,\ldots,\Agent_{k-2}\}+\Agent_n$ is matched in
$\Config'_k$ to its most preferred object in configuration
$\Config'_k$.  We begin by applying zero or one swaps to configuration
$\Config'_2$ to obtain configuration $\Config'_3$. If
$\Config'_2(\Agent_n)=\Object^*_1$, then we define $\Config'_3$ as
$\Config'_2$. If not, then $\Config'_2(\Agent_i)=\Object_1^*$ for some
$i$ in $\{2,\ldots,n-1\}$.  The preferences of agents $\Agent_i$ and
$\Agent_n$ ensure that a swap between these two agents can be applied
to configuration $\Config'_2$.  We define $\Config'_3$ as the
configuration that results from applying this swap. It easy to see
that configuration $\Config'_3$ belongs to $\Reach{\Config'}$ and
that every agent in $\Agents\cup\{\Agent_1,\Agent_n\}$ is matched in
$\Config'_3$ to its most preferred object under $\Prefs'$.

Now fix an integer $k$ such that $4\leq k\leq n+1$, and inductively
assume that we have constructed a configuration $\Config'_{k-1}$ in
$\Reach{\Config'}$ such that every agent in
$\Agents\cup\{\Agent_1,\ldots,\Agent_{k-3}\}+\Agent_n$ is matched to
its most preferred object under $\Prefs'$.  We apply zero or one swaps
to configuration $\Config'_{k-1}$ to obtain configuration
$\Config'_k$. If $\Config'_{k-1}(\Agent_{k-2})=\Object^*_{k-2}$, then
we define $\Config'_k$ as $\Config'_{k-1}$. If not, then
$\Config'_{k-1}(\Agent_i)=\Object_{k-2}^*$ for some $i$ in
$\{k-1,\ldots,n-1\}$.  The preferences of agents $\Agent_{k-2}$ and
$\Agent_i$ ensure that a swap between these two agents can be applied
to configuration $\Config'_{k-1}$. We define $\Config'_k$ as the
configuration that results from applying this swap. It easy to see
that configuration $\Config'_k$ belongs to $\Reach{\Config'}$ and
that every agent in
$\Agents\cup\{\Agent_1,\ldots,\Agent_{k-2}\}+\Agent_n$ is matched in
$\Config'_k$ to its most preferred object under $\Prefs'$.

Since $\Config'_{n+1}$ belongs to $\Reach{\Config'}$ and every agent
in $\Agents'$ is matched in $\Config'_{n+1}$ to its most preferred
object under $\Prefs'$, we conclude that $\Config'_{n+1}=\Config''$
and hence that $\Instance'$ is a positive instance of reachable
matching on general graphs.

Now we address the ``if'' direction.  Assume that instance
$\Instance'$ is a positive instance of reachable matching on general
graphs. Thus $\Config''$ belongs to $\Reach{\Config'}$, and hence
there is a sequence of swaps $S$ that transforms configuration
$\Config'$ into configuration $\Config''$.

By examining the preferences of the agents in $\Agents^*$, we deduce
that each agent in $\Agents^*$ participates in exactly one swap in
$S$, and that the other agent participating in each of these swaps
belongs to $\Agents$. By examining the preferences of the agents in
$\Agents$, we deduce that agent $\Agent_1$ is the agent that swaps
with agent $\Agent^*_n$, and that once an agent in $\Agents$ becomes
matched to an object in $\Objects^*$, it remains matched to an object
in $\Objects^*$ thereafter. It follows that there is a permutation
$\pi$ of $\{1,\ldots,n\}$ such that $\pi(n)=1$ and $\Agent^*_i$ swaps
with $\Agent_{\pi(i)}$ for $1\leq i\leq n$.

For any integer $k$ such that $0\leq k\leq |S|$, let $\Config'_k$
denote the configuration reached by applying the first $k$ swaps of
sequence $S$ to configuration $\Config'$. Thus $\Config'=\Config'_0$,
$\Config''=\Config'_{|S|}$, and $\Config'_k$ is of the form
$(\Oaf',\Matching'_k)$ where $\Matching'_k$ is
a perfect matching from $\Agents'$ to $\Objects'$.

For any integer $k$ such that $0\leq k\leq |S|$, we use the perfect matching
$\Matching'_k$ to construct a perfect matching $\Matching_k$ from $\Agents$ to
$\Objects$, as follows: for each agent $\Agent^*_i$ in $\Agents^*$
such that $\Matching'_k(\Agent^*_i)$ belongs to $\Objects^*$, we define
$\Matching_k(\Agent_{\pi(i)})$ as $\Matching'_k(\Agent_{\pi(i)})$; for each
agent $\Agent^*_i$ in $\Agents^*$ such that $\Matching'_k(\Agent^*_i)$
belongs to $\Objects$, we define $\Matching_k(\Agent_{\pi(i)})$ as
$\Matching'_k(\Agent^*_i)$.

For any integer $k$ such that $0\leq k\leq |S|$, we define $\Config_k$
as the configuration $(\Oaf,\Matching_k)$.  It is straightforward to
prove by induction on $k$ that $\Config_k$ belongs to
$\Reach{\Config}$ for $0\leq k\leq |S|$.

Let $\ell$ denote the least integer such that
$\Config^{-1}_{\ell}(\Object_n)=\Agent^*_n$. We know that $\ell$
exists since $\Config^{-1}_{|S|}(\Object_n)=\Agent^*_n$, and that
$\ell$ is positive since $\Config^{-1}_0(\Object_n)=\Agent_n$. As
discussed earlier, agent $\Agent_1$ is the only agent that
participates in a swap with agent $\Agent^*_n$. Hence
$\Config^{-1}_{\ell-1}(\Object_n)=\Agent_1$.  Since $\Config_{\ell-1}$
belongs to $\Reach{\Config}$, we conclude that $\Instance$ is a
positive instance of reachable object on generalized stars, as required.
\end{proof}

It is easy to see that the reachable matching on cliques problem belongs to NP.
We use a reduction from reachable object on cliques to reachable matching on cliques to establish that reachable matching on cliques is NP-complete.
This reduction is similar to the one used in the proof of Lemma~\ref{lem:NPC_RM_general}.

Fix an arbitrary reachable object on cliques instance $\Instance$. Without loss of
generality, we can assume that the associated configuration $\Config=
(\Oaf, \Matching)$ is such that $\Oaf =  (\Agents,\Objects,\Prefs,\Edges)$, 
$\Agents=\{\Agent_1,\ldots,\Agent_n\}$,
$\Objects=\{\Object_1,\ldots,\Object_n\}$, and
$\Config(\Agent_i)=\Object_i$ for $1\leq i\leq n$. We can also assume
without loss of generality that our goal is to determine whether there
is a matching $\Matching_1$ in $\Reach{\Config}$ such that
$\Matching_1(\Agent_1)=\Object_n$.

We now describe how to transform reachable object on cliques instance
$\Instance$ into a reachable matching on cliques instance
$\Instance'$.
Instance $\Instance'$ has two associated configurations 
$\Config' = (\Oaf', \Matching')$ and $\Config'' = (\Oaf', \Matching'')$, 
where $\Oaf' = (\Agents', \Objects', \Prefs', \Edges')$. 
The set of agents $\Agents'$ is equal to
$\Agents\cup\Agents^*$ where
$\Agents^*=\{\Agent^*_1,\ldots,\Agent^*_n\}$. The set of objects
$\Objects'$ is equal to $\Objects\cup\Objects^*$ where
$\Objects^*=\{\Object^*_1,\ldots,\Object^*_n\}$.
Let $K_{2n}$ denote
the complete graph with vertex set $\Objects'$, 
and let $\Edges'$ denote the edge set of $K_{2n}$.
The agent preferences $\Prefs'$ and the matchings $\Matching'$ and $\Matching''$ are as described in the proof of Lemma~\ref{lem:NPC_RM_general}.

Let $\hat{E}$ denote the union of three sets of edges: 
$\{(b_i, b_j) \mid i, j \in [n] \wedge i \neq j\}$; 
$\{(b_i, b_i^*) \mid i \in [n]\}$; 
$\{(b_i^*, b_j^*) \mid i, j \in [n] \wedge i \neq j\}$.
Lemma~\ref{lem:NPC_RM_clique_graphcorrect} below establishes that 
if a swap occurs on an edge $\Edge$ in $\Instance'$, 
then $\Edge$ belongs to $\hat{\Edges}$.

\begin{lemma}
\label{lem:NPC_RM_clique_graphcorrect}
Let $i$ and $j$ be elements of $[n]$ such that $i \neq j$.
Let $\Matching_1$ and $\Matching_2$ be matchings in $\Reach{\Config'}$ 
such that $\Swap{\Oaf}{\Matching_1}{\Matching_2}$.
Then $\Matching_2^{-1}(b_i) \neq \Matching_1^{-1}(b^*_j)$.
\end{lemma}
\begin{proof}
We consider two cases.

Case 1: $\Matching_1^{-1}(b^*_j)$ belongs to $A^*$. 
By examining the preferences of agents in $A^*$, we deduce that 
$\Matching_1^{-1}(b^*_j) = a_j^*$. 
The only object that agent $a_j^*$ prefers to $b^*_j$ is $b_j$.
Hence $\Matching_2^{-1}(b_i) \neq a_j^* = \Matching_1^{-1}(b^*_j)$. 

Case 2: $\Matching_1^{-1}(b^*_j)$ belongs to $A$. 
By examining the preferences of agents in $A$, we deduce that
$\Matching_2(\Matching_1^{-1}(b^*_j))$ belongs to $B^*$.
Hence $\Matching_2^{-1}(b_i) \neq \Matching_1^{-1}(b^*_j)$.
\end{proof}

Using Lemma~\ref{lem:NPC_RM_clique_graphcorrect}, 
along with the same reasoning as in the proof of Lemma~\ref{lem:NPC_RM_general}, we deduce that 
$\Instance'$ is a positive instance of reachable matching on cliques if and only if 
$\Instance$ is a positive instance of reachable object on cliques. 
Thus Theorem~\ref{thm:NPC_RM_clique} below holds.

\begin{theorem}
\label{thm:NPC_RM_clique}
    Reachable matching on cliques is NP-complete.
\end{theorem}



\subsection{NP-Hardness for Pareto-Efficiency on Cliques}
\label{sec:NP_PE_clique}

\begin{theorem}
\label{thm:NP_PE_clique}
Pareto-efficient matching on cliques is NP-hard.
\end{theorem}
\begin{proof}
We use the same reduction as we used in Section~\ref{sec:NPC_RM_clique} to establish the NP-completeness of reachable matching on cliques. 
In analyzing that reduction, we proved that a given instance of reachable object on cliques is positive if and only if every agent gets its most preferred object in the corresponding instance of reachable matching on cliques.
Therefore an efficient algorithm for computing a Pareto-efficient matching on cliques yields an efficient algorithm for reachable object on cliques.
Since reachable object on cliques is NP-complete, we deduce that Pareto-efficient matching on cliques is NP-hard.
\end{proof}
	\section{Other Polynomial-Time Bounds}
\label{sec:easy}

In this section, we briefly discuss simple algorithms that serve to
justify the remaining polynomial-time entries in
Table~\ref{tbl:dozen}.

For reachable matching on trees, the corresponding algorithm of
Gourv{\`e}s et al.~\cite{gourves2017object} for the object-moving
model can also be used for the agent-moving
model. In particular, for the agent-moving model, Section~
\ref{sec:RM_tree} shows that reachable matching problem on trees can be solved in $O(n^2)$ time. 		


For the other two polynomial-time table entries for stars,
observe that once an agent swaps away from the center object, it
cannot participate in another swap. This observation severely
restricts the swap dynamics, making it easy to establish that
the reachable object problem on stars can be solved in $O(n^2)$ time and
Pareto-efficient matching problem on stars can be solved in $O(n)$ time.

%

%

\subsection{Reachable Matching on Trees}\label{sec:RM_tree}
\begin{restatable}{theorem}{PolyRM}\label{thm:RM_star}
	Reachable matching	 on trees can be solved in $O(n^2)$ time.
\end{restatable}

\begin{proof}[Proof Sketch]
Let $\Matching$ and $\Matching'$ be the initial and target perfect
matchings in an instance of the reachable matching problem on a
tree. To solve the problem, we use the approach presented by
Gourv{\`e}s et al.~\cite{gourves2017object} for the object-moving
model. Observe that every agent $\Agent$ moves along a unique dipath
in the tree from $\Matching(\Agent)$ to
$\Matching'(\Agent)$. Therefore, it suffices to check whether there is
a sequence of swaps such that each agent moves according to its
dipath.  Let us say that an edge is good if the two agents currently
matched to the endpoints of the edge both need to cross the edge as
the next step along their respective dipaths.  Using essentially the
same argument as in the proof of Proposition~3
in Gourv{\`e}s et al.~\cite{gourves2017object}, we can show if the
target matching is reachable and is not equal to the current matching,
then at least one good edge exists. It follows easily that if the
target matching is reachable, then we can reach it by repeatedly
performing a swap across an arbitrary good edge.

We now discuss how to find good edges efficiently so that the overall
running time is $O(n^2)$.  We begin by traversing the tree to
construct a set containing all of the edges that are good in the
initial state.  The order of the edges within this set is immaterial,
so it can be implemented using a simple data structure such as a list
or stack.  Then, while the set of good edges is nonempty, we perform
the following steps. First, we remove an arbitrary good edge from the
set, and perform the associated swap. It is easy to see that the other
edges in the good set remain good after the swap. Furthermore, up to
two edges can become good as a result of the swap, and a simple local
search can be used to identify these edges in constant time. Any
newly-identified good edges are added to the set of good edges, and we
proceed to the next iteration.

Initialization of the set of good edges takes $O(n)$ time and each
iteration of the loop performs one swap and takes constant time. Since
the number of swaps is $O(n^2)$, the claimed time bound follows.
\end{proof}

\subsection{Algorithms for Stars}
\label{sec:algs_star}

In this section, we solve the reachable object and Pareto-efficient
matching problems on stars.  We refer to the object located at the
center of the star as the center object, and to the remaining objects
as leaf objects.  We refer to the agent that is currently matched to
the center object as the center agent $O$, and to the remaining agents as
leaf agents.

\subsubsection{Reachable Object on Stars}
\label{sec:RO_star}

\begin{restatable}{theorem}{PolyRO}\label{thm:RO_star}
	Reachable object on stars can be solved in $O(n^2)$ time. 
\end{restatable}
\begin{proof}[Proof Sketch]
Let $\Agent$ be the given agent and let $\Object$ the given target
object in an instance of the reachable object problem. Notice that
there is a unique dipath for $\Agent$ to follow to reach $\Object$,
and the dipath has at most two edges.
	
There are three cases for the dipath, from center to leaf, from leaf
to center, and from leaf to leaf. The center to leaf case is
straightforward as the only way for agent $\Agent$ to reach object
$\Object$ is to perform a single swap involving both of them. 

For the leaf to center case, we use the approach presented by
Gourv{\`e}s et al.~\cite{gourves2017object} for the object-moving
model.  The problem reduces to the search of a path in a digraph
$G=(\Agents, E)$ where $(a, a') \in E$ with $a\in \Agents, a'\in
\Agents\setminus\{O\}$ if and only if $a$ and $a'$ can rationally
trade when $a$ is at the center position and $a'$ is in its initial
position. There is a path from $O$ to $\Agent$ in $G$ if and only if
$\Agent$ can move to the center position. It takes $O(n^2)$ time to
construct $G$ and solve this path problem.

We reduce the leaf to leaf case to the leaf to center case, as
follows. First, we solve a leaf to center problem to determine whether
agent $\Agent$ can be moved to the center without moving the initial
owner of object $\Object$. If this is possible, then we check whether
agent $\Agent$ and the initial owner of object $\Object$ can trade
their objects. This procedure works correctly because the ownership of
any leaf object can change at most once in any valid sequence of
swaps. To see this, observe that once an agent moves from the center
to a leaf, it can never return to the center. Accordingly, for agent
$\Agent$ to move to object $\Object$, the initial owner of $\Object$
needs to remain stationary until $\Agent$ reaches the center.
\end{proof}

\subsubsection{Pareto Efficiency on Stars}
\label{sec:PE_star}

\begin{restatable}{theorem}{PolyPE}\label{thm:PE_star}
	Pareto-efficient matching on stars can be solved in $O(n)$ time. 
\end{restatable}
\begin{proof}[Proof Sketch]
We use an algorithm based on serial dictatorship
\cite{abdulkadirouglu1999house}.  First, we prune any leaf agent with
its associated object if the agent does not prefer the center object
to its own object, since this leaf agent will never be involved in any
swap.  Therefore, the center agent can swap with any remaining leaf
agent only if the leaf agent holds one of its preferred objects.  If
the top preference of the center agent out of the remaining objects is
the center object, then the current matching is Pareto-efficient as no
swaps can occur.  Otherwise, the top preference of the center agent
$\Agent$ out of the remaining objects is a leaf object $\Object$, and
we can apply the swap operation that moves agent $\Agent$ to object
$\Object$. The previous owner of object $\Object$ becomes the new
center agent.  Then, we prune agent $\Agent$ and object $\Object$ and
recurs on the new center agent. It is easy to see that this whole
process takes $O(n)$ time, and when each agent is pruned, it holds its
favorite object among the remaining objects. It follows that the
matching returned by this process is not Pareto-dominated by any other
matching, and hence is Pareto-efficient.
\end{proof}

	\section{Concluding Remarks}
\label{sec:concs}

In this paper, we have introduced the agent-moving model, and we have
revisited the collection of problems listed in Table~\ref{tbl:dozen},
which were previously considered in the context of the object-moving
model. In all cases where a polynomial-time algorithm or hardness
result has been established for the object-moving model, we have
established a corresponding result for the agent-moving model.

In addition, we have presented a polynomial-time algorithm for \PE\ on
generalized stars in the agent-moving model, a problem that remains
open in the object-moving model.  It is natural to ask whether our
techniques can be extended to address this open problem.  Our
algorithm relies on the polynomial-time solvability of the \RO\ problem
for the center agent, which allows us to compute an object that is
matched to the center agent in some Pareto-efficient matching.  In the
object-moving model, no polynomial-time algorithm is known to compute
an agent that is matched to the center object in some Pareto-efficient
matching.  (We do know how to compute the agents that can be reached by the center object in polynomial time, but it isn't clear how to use this information to
compute a Pareto-efficient matching in polynomial time.) An
interesting direction for future research in the agent-moving model is
to determine whether our techniques for solving \PE\ on generalized
stars can be extended to trees. It would also be interesting to study strategic aspects of this model.


	\bibliography{refs}

\begin{thebibliography}{10}

\bibitem{abdulkadirouglu1999house}
Atila Abdulkadiro{\u{g}}lu and Tayfun S{\"o}nmez.
\newblock House allocation with existing tenants.
\newblock {\em Journal of Economic Theory}, 88(2):233--260, 1999.

\bibitem{AgarwalEGV20}
Aishwarya Agarwal, Edith Elkind, Jiarui Gan, and Alexandros~A. Voudouris.
\newblock Swap stability in schelling games on graphs.
\newblock In {\em Proceedings of the 34th AAAI Conference on Artificial
  Intelligence}, pages 1758--1765, 2020.

\bibitem{aziz2016optimal}
Haris Aziz, P{\'a}ter Bir{\'o}, J{\'e}r{\^o}me Lang, Julien Lesca, and
  J{\'e}r{\^o}me Monnot.
\newblock Optimal reallocation under additive and ordinal preferences.
\newblock In {\em Proceedings of the 15th International Conference on
  Autonomous Agents and Multiagent Systems}, pages 402--410, 2016.

\bibitem{bentert2019good}
Matthias Bentert, Jiehua Chen, Vincent Froese, and Gerhard~J.\ Woeginger.
\newblock Good things come to those who swap objects on paths.
\newblock {\em arXiv:1905.04219}, 2019.

\bibitem{beynier2019local}
Aur{\'e}lie Beynier, Yann Chevaleyre, Laurent Gourv{\`e}s, Ararat Harutyunyan,
  Julien Lesca, Nicolas Maudet, and Ana{\"e}lle Wilczynski.
\newblock Local envy-freeness in house allocation problems.
\newblock {\em Autonomous Agents and Multi-Agent Systems}, 33(5):591--627,
  2019.

\bibitem{BiloBLM20}
Davide Bil{\`{o}}, Vittorio Bil{\`{o}}, Pascal Lenzner, and Louise Molitor.
\newblock Topological influence and locality in swap schelling games.
\newblock In {\em Proceedings of 45th International Symposium on Mathematical
  Foundations of Computer Science}, pages 15:1--15:15, 2020.

\bibitem{bilo2018almost}
Vittorio Bil{\`o}, Ioannis Caragiannis, Michele Flammini, Ayumi Igarashi,
  Gianpiero Monaco, Dominik Peters, Cosimo Vinci, and William~S.\ Zwicker.
\newblock Almost envy-free allocations with connected bundles.
\newblock In {\em Proceedings of the 10th Innovations in Theoretical Computer
  Science Conference}, pages 14:1--14:21, 2019.

\bibitem{chevaleyre2007allocating}
Yann Chevaleyre, Ulle Endriss, and Nicolas Maudet.
\newblock Allocating goods on a graph to eliminate envy.
\newblock In {\em Proceedings of the 22nd AAAI Conference on Artificial
  Intelligence}, pages 700--705, 2007.

\bibitem{endriss2006negotiating}
Ulrich Endriss, Nicolas Maudet, Fariba Sadri, and Francesca Toni.
\newblock Negotiating socially optimal allocations of resources.
\newblock {\em Journal of Artificial Intelligence Research}, 25:315--348, 2006.

\bibitem{gourves2017object}
Laurent Gourv{\`e}s, Julien Lesca, and Ana{\"e}lle Wilczynski.
\newblock Object allocation via swaps along a social network.
\newblock In {\em Proceedings of the 26th International Joint Conference on
  Artificial Intelligence}, pages 213--219, 2017.

\bibitem{huang2019object}
Sen Huang and Mingyu Xiao.
\newblock Object reachability via swaps along a line.
\newblock In {\em Proceedings of the 33rd AAAI Conference on Artificial
  Intelligence}, pages 2037--2044, 2019.

\bibitem{igarashi2019pareto}
Ayumi Igarashi and Dominik Peters.
\newblock Pareto-optimal allocation of indivisible goods with connectivity
  constraints.
\newblock In {\em Proceedings of the 33rd AAAI Conference on Artificial
  Intelligence}, pages 2045--2052, 2019.

\bibitem{david2013algorithmics}
David~F.\ Manlove.
\newblock {\em Algorithmics of Matching Under Preferences}.
\newblock World Scientific, 2013.

\bibitem{muller2020reachable}
Luis {M{\"u}ller} and Matthias {Bentert}.
\newblock On reachable assignments in cycles and cliques.
\newblock {\em arXiv:2005.02218}, 2020.

\bibitem{saad2009coalitional}
Walid Saad, Zhu Han, M{\'e}rouane Debbah, Are Hjorungnes, and Tamer Basar.
\newblock Coalitional game theory for communication networks.
\newblock {\em IEEE Signal Processing Magazine}, 26(5):77--97, 2009.

\bibitem{saffidine2018constrained}
Abdallah Saffidine and Ana{\"e}lle Wilczynski.
\newblock Constrained swap dynamics over a social network in distributed
  resource reallocation.
\newblock In {\em Proceedings of the 11th International Symposium on
  Algorithmic Game Theory}, pages 213--225, 2018.

\bibitem{sandholm1998contract}
Tuomas Sandholm.
\newblock Contract types for satisficing task allocation.
\newblock In {\em Proceedings of the AAAI Spring Symposium: Satisficing
  Models}, pages 23--25, 1998.

\bibitem{shapley1974cores}
Lloyd Shapley and Herbert Scarf.
\newblock On cores and indivisibility.
\newblock {\em Journal of Mathematical Economics}, 1(1):23--37, 1974.

\bibitem{Yoshinaka05}
Ryo Yoshinaka.
\newblock Higher-order matching in the linear lambda calculus in the absence of
  constants is {NP}-complete.
\newblock In {\em Proceedings of 16th International Conference on Rewriting
  Techniques and Applications}, pages 235--249, 2005.

\end{thebibliography}
	

\end{document}